\documentclass[Journal]{IEEEtran}

\usepackage[normalem]{ulem}

\usepackage{cite}
\usepackage[linesnumbered, algo2e, ruled,norelsize]{algorithm2e}
\usepackage{amsmath,amssymb,amsfonts}
\usepackage{algorithmic}
\usepackage{graphicx}
\usepackage{textcomp}
\usepackage{graphicx}
\usepackage{dsfont}
\usepackage{caption}
\usepackage{subcaption}
\usepackage{latexsym}
\usepackage{amsmath}
\usepackage[mathscr]{euscript}
\usepackage[linesnumbered, algo2e, ruled,norelsize]{algorithm2e}
\SetArgSty{textnormal}
\usepackage{threeparttable}
\usepackage{threeparttablex}
\usepackage{commath}
\usepackage{mathtools}
\usepackage{pict2e}
\usepackage{epstopdf}
\usepackage{verbatim}
\usepackage{subfloat}
\usepackage{tabu}
\usepackage{booktabs}
\usepackage{float}
\usepackage{amssymb}
\usepackage{url}
\usepackage{multirow}
\usepackage{graphicx}
\usepackage{mathptmx}
\usepackage{cite}
\usepackage{footnote}
\usepackage{array}
\usepackage[table]{xcolor}
\usepackage{tabularx}
\usepackage{ulem}
\usepackage{array}
\usepackage{lipsum}
\usepackage{stfloats}

\usepackage{amsthm} 
\newtheorem{Proposition}{\bf{Proposition}}
\makeatletter
\newcommand{\leftlabel}[1]{&&
  \refstepcounter{equation}\ltx@label{#1}%
  \tagform@{\theequation}&&}
\makeatletter

\def\BibTeX{{\rm B\kern-.05em{\sc i\kern-.025em b}\kern-.08em
    T\kern-.1667em\lower.7ex\hbox{E}\kern-.125emX}}

\begin{document}

\title{Practical Design of Probabilistic Constellation Shaping for Physical Layer Security in Visible Light Communications}

\author{Thanh~V.~Pham,~\IEEEmembership{Member,~IEEE,}
        and Susumu Ishihara, ~\IEEEmembership{Member,~IEEE}
\thanks{Thanh V. Pham and Susumu Ishihara are with the Department of Mathematical and Systems Engineering, Shizuoka University, Shizuoka, Japan (e-mail: pham.van.thanh@shizuoka.ac.jp, ishihara.susumu@shizuoka.ac.jp). 

Part of this paper has been presented at the 2024 IEEE Wireless Communications and Networking Conference (2024 IEEE WCNC), Workshop on Optical Wireless Communications, Dubai, UAE, April 2024 \cite{Pham2024}. This work is supported by the Japan Society for Promotion of Science (JSPS) KAKENHI under Grant 23K13333.}} 
\maketitle
\begin{abstract}
 This paper studies a practical design of probabilistic constellation shaping (PCS) for physical layer security in visible light communications (VLC). In particular, we consider a wiretap VLC channel employing a probabilistically shaped $M$-ary pulse amplitude modulation (PAM) constellation. Considering the requirements for reliability of the legitimate user's channel, flickering-free transmission, and symmetric constellation distribution, the optimal constellation distributions to maximize modulation-constrained secrecy capacity or the bit error rate (BER) of eavesdropper's channel are investigated for both scenarios of known and unknown eavesdropper's channel state information (CSI). To formulate the constraint on the channel reliability, tractable closed-form expressions for the upper bound and approximate BER of $M$-ary PAM under an arbitrary symbol probability are derived. The design problem is shown to be non-convex due to the non-convex BER constraint. By proving that the upper bound BER is a concave function of the constellation distribution, a suboptimal solution based on the convex-concave procedure (CCCP) is presented. Our findings reveal that while the uniform signaling can only satisfy the BER constraint when the optical power is beyond a certain value, the proposed PCS design works in the entire region of the optical power. Some insights into the optimal constellation distribution with respect to the emitted optical power are also discussed. 

\end{abstract}

\begin{IEEEkeywords}
Visible light communications, probabilistic constellation shaping, physical layer security. 
\end{IEEEkeywords}


\section{Introduction}
\subsection{Background}
The past decade has seen tremendous growth in the number of mobile devices and data-intensive applications. This inevitably poses a serious burden on the existing wireless systems, which, at the same time, are suffering from the spectrum shortage problem \cite{Li2024}. In this regard, visible light communications (VLC) has emerged as an attractive technology to complement current radio frequency (RF) systems. Operating at the visible light frequencies, VLC offers several distinct advantages, such as high data rates, license-free spectrum, and no interference with co-existing RF systems \cite{Matheus2019}. 

Like conventional wireless systems, VLC is more susceptible to eavesdropping than wired communications due to the broadcast nature of the visible light signal. Traditional approaches against eavesdroppers are performed at the upper layers of the network protocol stack (i.e., network layer and above), primarily by means of key-based cryptography. Commonly used cryptographic algorithms such as the Advanced Encryption Standard (AES) and RSA are considered computationally secure as breaking the secret key within a reasonable amount of time is almost impossible with the current computing power. Nonetheless, the rapid development of quantum computers, whose computational capability is exponentially greater than their classical counterparts, presents a foreseeable threat to conventional security measures. Therefore, there has been great interest in information-theoretic security, including physical layer security (PLS), which can deal with eavesdroppers with unlimited computing resources \cite{Wyner1975}. A central issue of PLS is to characterize the so-called secrecy capacity, which defines the maximum transmission rate between a transmitter and a legitimate user satisfying that an eavesdropper cannot decode any information from its received signal. 
\subsection{Related Work}
Research on PLS in VLC systems predominantly focuses on the multiple-input single-output (MISO) configuration due to its more prevalent in practice as multiple LED luminaires are often deployed to provide sufficient illumination. Exploiting the spatial degrees of freedom provided by multiple LED transmitters, a considerable number of studies examined the use of precoding and artificial noise (AN) to improve the PLS performance. Specifically, various precoding designs have been investigated under different system configurations, secrecy constraints, and objectives using classical optimization \cite{mostafa2015physical, mostafa2016optimal, ma2016optimal, pham2017secrecy, arfaoui2018secrecy, Cho2020, SonDuong2021} or machine learning techniques \cite{Xiao2019, Hoang2024}. The use of AN (also known as friendly jamming) has also been explored in \cite{Shen2016,Wang2018OpticalJaming,Pham2018,Cho2019, Arfaoui2019,pham2020energy,Pham2024-journal}, which verified that AN-aided precoding was superior to precoding-only approach in enhancing the secrecy performance. Aside from precoding and AN, generalized space-shift keying (GSSK) with optimal LED pattern selection has also been employed as a PLS technique \cite{WangFasong2018}.    

It is worth mentioning that while studies on the PLS for MISO VLC systems were substantial, there were, however, sporadic works on the conventional single-input single-output (SISO) wiretap VLC channel, which involves one LED transmitter. Unlike RF systems, where the signal can be complex and is often constrained by an average power, the signal in VLC systems must be real, non-negative, and is usually subject to a peak power constraint due to the inherent limited linear range of the LEDs and/or the eye-safety regulations \cite{Soltani2022}. As a result, an exact, simple closed-form expression for the secrecy capacity of the VLC channels is still not known. Given the difficulty of deriving an exact closed-form formula, upper and lower bounds on the secrecy capacity were provided in \cite{Wang2018} and \cite{Wang2023} for the case of signal-independent and signal-dependent noise models, respectively. 
Moreover, it was shown in \cite{Ozel2015} that the secrecy capacity-achieving distribution of amplitude-constrained (i.e., peak power-constrained) optical wireless channels is discrete with a finite number of mass points. 

While previous studies on PLS for VLC systems predominantly focused on characterizing or optimizing the information-theoretic (lower/upper bound) secrecy capacity (i.e., maximum secrecy rate that can be reliably transmitted irrespective of modulation and hardware constraints), very little work investigated the issue of modulation-constrained secrecy capacity with communication reliability of the legitimate user's channel, and illumination quality (e.g., flickering-free), which are integral aspects of practical VLC systems. A recent study in \cite{Valeria2024} proposed a bias-hoping technique to enhance the physical secrecy performance of SISO VLC channels employing $M$-ary pulse amplitude modulation ($M$-PAM). Particularly, the DC bias is varied symbol by symbol in such a way that it degrades the performance of symbol detection at the eavesdropper, leading to its sufficiently high bit error rate (BER). 
Our works in \cite{Hoang2024} and \cite{Hoang2024-ICCE} studied the use of reinforcement learning and particle swarm optimization (PSO) for joint optimization of precoding and $M$-PAM modulation order to maximize the modulation-constraint secrecy capacity of MISO VLC channels considering the reliability of the legitimate user's channel and the unreliability of the eavesdropper's channel.   

\subsection{Our Contributions}
The implication of the result in \cite{Ozel2015} is that for an amplitude-based modulation scheme, the location of the constellation points (i.e., symbol amplitude) and their transmission probability should be jointly optimized to achieve the secrecy capacity. In literature, the optimizations of location and probability of the constellation points are respectively known as \textit{geometric constellation shaping (GCS)} and \textit{probabilistic constellation shaping (PCS)}. It is, however, noted that the possible irregularity of the constellation points induced by GCS would increase the complexity of the digital signal processing (DSP), thus rendering GCS impractical for commercial systems \cite{Junho2019}. This leaves PCS as a feasible approach to improve the secrecy performance. Although PLS approaches considering $M$-PAM were previously studied, e.g., in \cite{Hoang2024} and \cite{Valeria2024}, these works assumed the uniform distribution of the constellation symbols, i.e., no constellation shaping. 
{\textit{A limitation of uniform signaling in VLC systems is that when the emitted optical power is not sufficiently high (e.g., due to dimming requirements or to comply with the eye safety regulations), the BER of the legitimate user's channel can be unacceptably high, rendering the channel unreliable}}.   
In this paper, we, therefore, propose practical designs of PCS for SISO VLC wiretap channels, which aim at maximizing the modulation-constrained secrecy capacity or the BER of the eavesdropper's channel given constraints on the communication reliability of the legitimate user's channel, flickering-mitigated transmission, and symmetric constellation distribution. The main contributions of this work are summarized as follows. 
\begin{itemize}
    \item To incorporate the constraint of communication reliability into the optimal PCS design problem, we study the symbol error rate (SER) of $M$-PAM given an arbitrary distribution of the constellation symbols. Since an exact SER expression remains unknown,  tractable upper bound and approximation are presented in closed-form expressions. The tightness of the derived bound and approximation is compared to the exact SER with different settings of constellation distribution. Furthermore, constraints to fulfill the requirements of flickering-free transmission and symmetric constellation distribution are also introduced. 
    \item PCS design approaches that maximize the modulation-constrained secrecy capacity for the case of known and unknown eavesdropper's channel state information (CSI) at the transmitter are investigated. In the case of the unknown eavesdropper's CSI, an average secrecy capacity over the eavesdropper's CSI is considered, which, however, does not lead to a tractable mathematical expression. Therefore, a simple lower bound is presented and utilized for optimization. Moreover, we study a PCS design based on the Quality-of-Service (QoS) perspective that aims to maximize the BER of the eavesdropper's channel. 
    \item We prove that the upper bound and approximation BER expressions are concave functions of the constellation distribution. As a result, the optimal design problems are shown to be not convex due to the non-convexity of the communication reliability constraint (i.e., the constraint that the upper-bound BER of the legitimate user's channel is below a predefined threshold). The concavity of the upper bound and approximate BER also enables the use of the convex-concave procedure (CCCP), which guarantees solving local optima of non-convex problems with reasonable complexity \cite{yuille2003}.

\end{itemize}

\subsection{Organization}
The rest of the paper is structured as follows. The wiretap VLC system and channel model are described in Section II. The modulation-constrained secrecy capacity of PCS $M$-PAM is derived in Section III. Section IV presents the proposed practical PCS design, which takes into consideration the constraints on the BER of the legitimate user's channel, flickering-free transmission, and symmetric constellation distribution. Solutions to the proposed PCS designs are given in Section V  for the case of known and unknown eavesdropper's CSI. 
Simulation results are discussed in Section VI, and the paper is concluded in Section VII. 

{\it{Notation}}: The following notations are used throughout the paper. Uppercase and lowercase letters (e.g.,  $\mathbf{H}$ and $\mathbf{h}$) represent matrices and colum, n vectors, respectively, with $\mathbf{h}^T$ being the transpose of $\mathbf{h}$. Moreover, $\mathbf{0}_{N}$ and $\mathbf{1}_{N}$ are the all-zero and all-one vectors of size $N$. $\mathbb{I}(\cdot;\cdot)$ and $\mathbb{H}(\cdot)$ denote the mutual information and differential entropy, respectively. Also, $\mathbb{E}[\cdot]$ is the expected value operation, $|\cdot|$ is the absolute value.

\section{System and Channel Model}
\label{system}
\begin{figure}[t]
    \centering
    \includegraphics[width = 7.5 cm, height = 5.2cm]{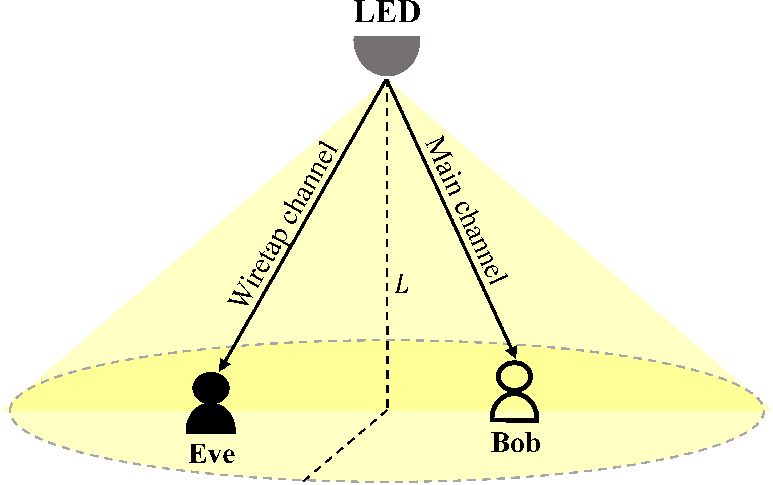}
    \caption{SISO VLC wiretap channel.}
    \label{Fig0}
\end{figure}
As depicted in Fig.~\ref{Fig0}, we examine a VLC wiretap channel consisting of an LED transmitter, a legitimate user Bob, and an eavesdropper Eve. Assume that Bob and Eve are both equipped with a single photodiode (PD) receiver and the LED is $L$ (m) above the receiver plane. 
Let $h_{\text{U}}$ be the channel gain between the LED and the user U\footnote{In the subsequent analysis, when necessary, the subscript `U' is replaced by `B' or `E' to refer to Bob or Eve, respectively.}. For the LED transmitter, let us also assume a Lambertian beam distribution whose emission intensity is given by
\begin{align}
    \mathcal{L}(\phi) = \frac{l+1}{2\pi}\cos^l(\phi),
\end{align}
where $\phi$ is the angle of irradiance and $l = \frac{-\log(2)}{\log(\cos(\Phi_{1/2}))}$ with $\Phi_{1/2}$ being the LED's semi-angle at half
power is the order of Lambertian emission. 
According to \cite{Komine2004}, $h_{\text{U}}$ is given by 
\begin{align}
 h_\text{U} =
  \begin{cases}
    \frac{A_{\text{r}}}{d^2_{\text{U}}}\mathcal{L}(\phi)T(\psi_{\text{U}})g(\psi_{\text{U}})\cos(\psi_{\text{U}})     & \quad 0 \leq \psi_{\text{U}} \leq \Psi_{\text{U}},\\
    0  & \quad \psi_{\text{U}} > \Psi_{\text{U}},
  \end{cases}
  \label{channel_gain}
\end{align}
where $A_r$ is the area of the PD, $d_{\text{U}}$ is the LoS link length, $\psi_{\text{U}}$ is the angle of incidence, and $\Psi_{\text{U}}$ is the field of view (FoV) of the PD. Here, $T(\psi_{\text{U}})$ is the gain of the optical filter and $g(\psi_{\text{U}})$ is the gain of the optical concentrator. Denoting $\kappa$ as the refractive index of the concentrator,  $g(\psi_{\text{U}})$ is given by
\begin{align}
 g(\psi_{\text{U}}) =
  \begin{cases}
    \frac{\kappa^2}{\sin^2(\Psi_{\text{U}})}    & \quad 0 \leq \psi_{\text{U}} \leq \Psi_{\text{U}},\\
    0  & \quad \psi_{\text{U}} > \Psi_{\text{U}}.
  \end{cases}
\end{align}
\section{Modulation-constrained Secrecy Capacity Analysis}
A probabilistically shaped constellation can be generated by a distribution matcher, which maps a uniform bitstream into the desired symbol distribution. A commonly used distribution matcher is the constant composition distribution matcher (CCDM), which is invertible, fixed-to-fixed length mapping, and has low complexity \cite{Schulte2016}. Let $s_1$, $s_2$, $\hdots$, $s_M$ be the  $M$ equally-spaced bipolar $M$-PAM symbols generated from a CCDM and denote $a_m$, $p_m$ as the amplitude and transmission probability of $s_m$. Let us also assume that the symbol peak amplitude is $A$, that is, $|a_m| \leq A$. Due to the symbol symmetry around 0,   $a_m = (2m-M-1)\frac{A}{M-1}$ for $m = 1,~\hdots,~M$. 

Assuming $s$ be the transmitted symbol. To satisfy the nonnegative constraint, a DC bias $I_{\text{DC}}$ is added to the data symbols, resulting in the LED's drive current as
\begin{align}
    x = s + I_{\text{DC}}.
    \label{drive_current}
\end{align}
It is well-known that for an LED, there exists a dynamic linear range over which the output optical power is proportional to the input drive current. Let $I_{\text{min}}$ and $I_{\text{max}}$ be the lower and upper limit of the dynamic linear range. Constraining \eqref{drive_current} between $I_{\text{min}}$ and $I_{\text{max}}$ results in 
\begin{align}
    |s| \leq \Delta_{\text{DC}}, 
\end{align}
where $\Delta_{\text{DC}} = \min\left(I_{\text{max}} - I_{\text{DC}},~I_{\text{DC}} - I_{\text{min}}\right)$. Therefore, it is reasonable to assume  that the peak amplitude $A$ is set to $A  = \Delta_{\text{DC}}.$

If $\eta$ is denoted as the electrical-to-optical conversion factor, the LED's output optical power is $P_{\text{t}} = \eta x$. The received electrical signal at the receiver $U$ is then given by
\begin{align}
    y_{\text{U}} & = \gamma h_{\text{U}}P_{\text{t}} + n_{\text{U}} 
                  = h_{\text{U}}\gamma\eta(s + I_{\text{DC}}) + n_{\text{U}},
    \label{received-current-signal}
\end{align}
where $\gamma$ is the responsivity of the PD and $n_{\text{U}}$ is assumed to be a zero-mean additive white Gaussian noise with variance $\sigma^2_{\text{U}}$. Here, 
\begin{align}
    \sigma^2_{\text{U}} = B_{\text{mod}}\left(2e\gamma h_{\text{U}}P_{\text{t}} + 4\pi e \gamma A_r\chi_{\text{amb}}\left(1-\cos\Psi_{\text{U}}\right) + i^2_{\text{amp}}\right),
\end{align}
where $B_{\text{mod}}$ is the modulation bandwidth, $e$ is the elementary charge, $\chi_{\text{amb}}$ is the ambient light photocurrent, and $i_{\text{amp}}$ is the pre-amplifier noise density current \cite{zeng2009high}. 

In \eqref{received-current-signal}, the DC bias carries no information, which thus can be removed for symbol detection. The filtered signal is then written by
\begin{align}
    \overline{y}_{\text{U}} = h_{\text{U}}\gamma\eta s + n_{\text{U}}. 
    \label{received-current-signal-filtered-out}
\end{align}
The capacity of the above channel is given as the mutual information between $s$ and $\overline{y}_{\text{U}}$, which is
\begin{align}
    \mathbb{I}(s; \overline{y}_{\text{U}}) & = \mathbb{H}(\overline{y}_{\text{U}}) - \mathbb{H}(\overline{y}_{\text{U}}|s) 
     = \mathbb{H}(\overline{y}_\text{U}) - \mathbb{H}(n_{\text{U}}) \nonumber \\
    & = -\int_{-\infty}^\infty p(\overline{y}_{\text{U}})\log_2 p(\overline{y}_{\text{U}})\text{d}\overline{y}_{{\text{U}}}- \frac{1}{2}\log_2(2\pi e\sigma^2_{\text{U}}),
    \label{mutual_information}
\end{align}
where $p(\overline{y}_{\text{U}})$ denotes the probability density function (PDF) of $\overline{y}_{\text{U}}$. Due to \eqref{received-current-signal-filtered-out} and the assumption that $n_\text{U}$ is Gaussian, $p(\overline{y}_{\text{U}})$ can be written as
\begin{align}
    p(\overline{y}_{\text{U}}) & = \sum_{m = 1}^M p\left(\overline{y}_\text{U} | s = s_m\right)\text{Pr}(s = s_m) \nonumber \\
    & = \frac{1}{\sqrt{2\pi\sigma^2_{\text{U}}}}\sum_{m=1}^M p_m \exp\left(-\frac{\left(\overline{y}_\text{U} - r_{\text{U}, m}\right)^2}{2\sigma^2_\text{U}}\right),
    \label{mixedGaussianPDF}
\end{align}
where $r_{\text{U}, m} = h_\text{U}\gamma\eta a_m$. 
Since the secrecy capacity is the difference between the capacities of Bob's and Eve's channels, it is given by
\begin{align}
    C_{\text{s}}(\mathbf{p})  =  & -\int_{-\infty}^\infty p(\overline{y}_{\text{B}})\log_2 p(\overline{y}_{\text{B}})\text{d}\overline{y}_{{\text{B}}} + \int_{-\infty}^\infty p(\overline{y}_{\text{E}})\log_2 p(\overline{y}_{\text{E}})\text{d}\overline{y}_{{\text{E}}} \nonumber \\ 
   & + \frac{1}{2}\log_2\left(\frac{\sigma^2_\text{E}}{\sigma^2_\text{B}}\right), 
   \label{secrecy_capacity}
\end{align}
where $\mathbf{p} = \begin{bmatrix}p_1 & p_2 & \cdots & p_M
\end{bmatrix}^T$ is used for mathematical convenience in the later parts of the paper. In this work, we are interested in the case of a positive secrecy capacity, which is attained if and only if $\frac{h_{\text{B}}}{\sigma_{\text{B}}} > \frac{h_{\text{E}}}{\sigma_{\text{E}}}$ holds. In other words, the signal observed by Eve is a degraded version of that received by Bob. 
\section{Practical PCS Design}
Given the secrecy capacity in \eqref{secrecy_capacity}, a straightforward optimal design of the probability vector $\mathbf{p}$ can be as follows 
\begin{subequations}
\label{OptProb1}
    \begin{alignat}{2}
        &\underset{\mathbf{p}}{\text{maximize}} & \hspace{2mm} & C_{\text{s}}(\mathbf{p}) \label{obj1}\\
        &\text{subject to }  &  & \nonumber \\
        & & & 0 \leq \mathbf{p} \leq 1, \label{constraint11} \\
        & & & \mathbf{1}^T_M \mathbf{p} = 1. \label{constraint12}
    \end{alignat}
\end{subequations}
Under the assumption that $\frac{h_{\text{B}}}{\sigma_{\text{B}}} > \frac{h_{\text{E}}}{\sigma_{\text{E}}}$, $C_{\text{s}}(\mathbf{p})$ is a concave function of $\mathbf{p}$ \cite{Dijk1997}. The problem in \eqref{OptProb1} is, therefore, a convex optimization problem, which can be solved efficiently using available software \cite{cvx}.  

However, the problem design in \eqref{OptProb1} ignores the following three practically important issues.
\begin{enumerate}
\item \textbf{Communication Reliability}: Recall that secrecy capacity is the maximum rate at which information can be reliably and secretly transmitted from Alice to Bob. In practical communications systems, communication reliability is guaranteed by advanced forward error-correction codes (FEC). Assuming the use of hard decision (HD)-FEC, a certain BER before FEC decoding (referred to as pre-FEC BER) must be achieved so that the BER after FEC decoding (i.e., post-FEC BER) can be made sufficiently low to attain channel reliability. Let $P^{\text{U}}_{\text{e}} (\mathbf{p})$ be the pre-FEC BER at the user $\text{U}$. A reliable transmission implies $   P^{\text{U}}_{\text{e}}(\mathbf{p}) \leq {P}^{\text{pre}}_{\text{e}}, $
with ${P}^{\text{pre}}_{\text{e}}$ being the upper limit for the pre-FEC BER. Apparently, this requirement was not considered in \eqref{OptProb1}.
\item \textbf{Flickering}: In VLC, the communication functionality is secondary to the illumination requirement. That is, any negative effects on the illumination caused by the communications must be mitigated as much as possible. In this aspect, the second shortcoming of \eqref{OptProb1} is that a constraint on channel-induced flickering was not considered. To see this, let us first examine the case of uniform symbol transmission where $\mathbb{E}[s] = 0$. Following \eqref{drive_current}, the average LED's emitted optical power is given by $\mathbb{E}[P_t] = \eta I_{\text{DC}}$, which is independent of the solution to \eqref{OptProb1}. However, in the case of PCS, one can see that $\mathbb{E}[P_t] = \eta (I_{\text{DC}} + \mathbf{a}^T\mathbf{p})$ (where $\mathbf{a} = \begin{bmatrix}a_1 & a_2 & \cdots & a_M
\end{bmatrix}^T$), which is generally dependent of $\mathbf{p}$. As the solution to \eqref{OptProb1} typically changes with the channels ${h}_\text{B}$ and ${h}_{\text{E}}$, the average LED's emitted optical power varies with the channels as well. This inevitably raises the problem of flickering when channels change due to users' movements. 
\item \textbf{Symmetric Constellation Distribution}: PCS can be practically enabled by the probabilistic amplitude shaping (PAS) architecture \cite{Georg2015}. However, the PAS assumes that the shaped constellation distribution is symmetric around zero. Obviously, the solution to \eqref{OptProb1} when constraints on communication reliability and flickering are considered may not satisfy this requirement.  
\end{enumerate}
In the following, we address the above-mentioned issues by incorporating additional constraints on the constellation distribution to \eqref{OptProb1}.  
\subsection{Bit-Error Rate of PCS-PAM}
Characterizations of the BER and SER of $M$-PAM are classical problems in communication theory. Nonetheless, most studies in the literature assume the uniform distribution of the constellation symbol. This assumption is often justifiable in practice since the source bits are generally uniformly distributed. Assuming the Gaussian noise, a closed-form expression of the SER can be easily derived by determining the symbol decision regions via the maximum likelihood (ML) detection rule \cite{proakis2008digital}. In a more general case when the symbol distribution is not uniform, 
the maximum a posteriori (MAP) detector archives the optimal performance. The detection rule for the channel in \eqref{received-current-signal-filtered-out}, according to the MAP detector, is given by \cite{proakis2008digital}
\begin{align}
    \hat{m} = \underset{1 \leq m \leq M}{\text{arg max}}~~p_mp(\overline{y}_{\text{U}}|s = s_m). 
    \label{MAP-detector}
\end{align}
The decision region of the symbol $s_m$ is thus defined by
\begin{align}
    D_m = \left\{\overline{y}_{\text{U}}\in \mathbb{R}:~ p_mp(\overline{y}_{\text{U}}|s=s_m) > p_{n}p(\overline{y}_{\text{U}}|s=s_{n})\right\},
\end{align}
for all $1 \leq n \leq M$ and $n \neq m$, 
which can be simplified  as
\begin{align}
    D_m = \left\{\overline{y}_{\text{U}}\in \mathbb{R}:~2\overline{y}_{\text{U}}(r_{\text{U}, m} - r_{\text{U}, n}) > 2\sigma^2_{\text{U}}\log\frac{p_{n}}{p_m} + r^2_{\text{U}, m} - r^2_{\text{U},n}\right\}.
    \label{decision-region}
\end{align}
To derive the exact SER, one must specifically characterize $D_m$, which, however, can be very challenging due to the generality of the distribution $\mathbf{p}$. In very recent work, exact formulas for the SERs of $M$-PAM and $M$-ary quadrature amplitude modulation ($M$-QAM) are derived where the symbol probability is drawn from the Maxwell-Boltzmann distribution, which maximizes the channel capacity under an average symbol power constraint \cite{Gutema2023}. It should be noted that the assumption of the Maxwell-Boltzmann distribution allows explicit characterizations of the decision regions, hence explicit expressions for the SER. However, the Maxwell-Boltzmann distribution may not be optimal when the objective is to maximize the secrecy capacity under the BER, flickering, and distribution symmetry constraints. To the best of our knowledge, while a simple closed-form expression for the exact SER has not yet been available in the general case, several lower and upper bounds with different degrees of tightness and complexity have been proposed \cite{Caen1997, Kuai2000, Mao2013}. It should be noted that the derived bounds are typically based on evaluating the probability of a union of events whose upper bound is generally tractable. Nonetheless, the proposed lower bounds are more complicated and are not given in closed-form expressions. Therefore, to facilitate the optimal PCS designs, we first focus on deriving tractable closed-form upper bound and approximation formulas for the SER of $M$-PAM constellation with arbitrary distribution.  
\subsubsection{Upper Bound}
To derive an upper bound for the SER, let $E_{m, n}$ be the event that the transmitted symbol $s_m$ is detected as symbol $s_n$ ($n \neq m$) and $P_{m,n}(\mathbf{p})$ be the probability of $E_{m,n}$. An upper bound for the SER is then given by
\begin{align}
    P_s^{\text{U}}(\mathbf{p}) & = \sum_{m = 1}^Mp_m\text{Pr}\left(\bigcup\limits_{\substack{n = 1 \\ n\neq m}}^{M}E_{m,n}\right)  \leq \sum_{m=1}^M p_m\sum_{\substack{n = 1 \\ n\neq m}}^M P^{\text{U}}_{ m,n}(\mathbf{p}),
\end{align}
where, according to the detection rule in \eqref{MAP-detector}, we have
\begin{align}
    P^{\text{U}}_{m, n}(\mathbf{p}) = \text{Pr}\left(p_mp(\overline{y}_{\text{U}}|s = s_m) \leq p_np(\overline{y}_{\text{U}}|s = s_n)\right).
    \label{SEP1}
\end{align}

\begin{Proposition}
A closed-form expression for \eqref{SEP1} is given by
\begin{align}
    P^{\rm{U}}_{m, n}(\mathbf{p}) = \frac{1}{2}{\rm{erfc}}\left(\frac{2\sigma^2_{\rm{U}}\log\frac{p_m}{p_n} + d^2_{{\rm{U}}, m, n}}{2\sqrt{2}\sigma_{{\rm{U}}}|d_{{\rm{U}}, m, n}|}\right),
    \label{closed-form-SEP}
\end{align}
where $d_{{\rm{U}}, m,n} = h_{{\rm{U}}}\gamma\eta(a_m - a_{n})$. 
\end{Proposition}

\begin{proof}
    The proof is given in Appendix A. 
\end{proof}
\noindent Assuming the Gray coding, the upper bound BER is thus given by
\begin{align}
    P^{\text{U}}_{\text{e,ub}}(\mathbf{p}) = \frac{1}{2\log_2M}\sum_{m = 1}^M\sum_{\substack{n = 1 \\ n\neq m}}^Mp_m{\rm{erfc}}\left(\frac{2\sigma^2_{\rm{U}}\log\frac{p_m}{p_n} + d^2_{{\rm{U}}, m, n}}{2\sqrt{2}\sigma_{{\rm{U}}}|d_{{\rm{U}}, m, n}|}\right).
    \label{PCS-BER}
\end{align}
\subsubsection{Approximation}
It is intuitive that for the transmitted symbol $s_m$, as the signal-to-noise ratio (SNR) increases or when $p_m$ is not much different, it is more likely to be wrongly detected as its adjacent symbol $s_{m-1}$ or $s_{m+1}$ (i.e., its closest symbols) than other symbols. That is $\text{Pr}\left(\bigcup\limits_{\substack{n = 1 \\ n\neq m}}^{M}E_{m, n}\right) \approx P^{\text{U}}_{ m,m-1}(\mathbf{p}) + P^{\text{U}}_{m, m+1}(\mathbf{p})$. As a result, a simple approximate expression of the SER can be given by
\begin{align}
    P_{s, \text{ap}}^{\text{U}}(\mathbf{p}) & =  \sum_{m=1}^M p_m\left(P^{\text{U}}_{ m,m-1}(\mathbf{p}) + P^{\text{U}}_{m, m+1}(\mathbf{p})\right) \nonumber \\
    & = \sum_{m = 1}^M \frac{p_m}{2}\left({\rm{erfc}}\left(\frac{2\sigma^2_{\rm{U}}\log\frac{p_m}{p_{m-1}} + d^2_{{\rm{U}}, m, m-1}}{2\sqrt{2}\sigma_{{\rm{U}}}|d_{{\rm{U}}, m, m-1}|}\right) \right. \nonumber \\ 
    & \left. \hspace{15mm}+ {\rm{erfc}}\left(\frac{2\sigma^2_{\rm{U}}\log\frac{p_m}{p_{m+1}} + d^2_{{\rm{U}}, m, m+1}}{2\sqrt{2}\sigma_{{\rm{U}}}|d_{{\rm{U}}, m, m+1}|}\right) \right).
    \label{approximate-SER}
\end{align}
In the above expression, we define $p_0 = p_{M+1} = 0$ and $d_{\text{U}, 1, 0} = \infty$, $d_{\text{U}, M, M+1} = -\infty$ since there is no symbol at the left to $s_1$ and no symbol at the right to $s_M$. 
\begin{figure*}[ht]
     \centering
     \begin{subfigure}[b]{0.48\textwidth}
         \centering
         \includegraphics[width=\textwidth, height = 5.2cm]{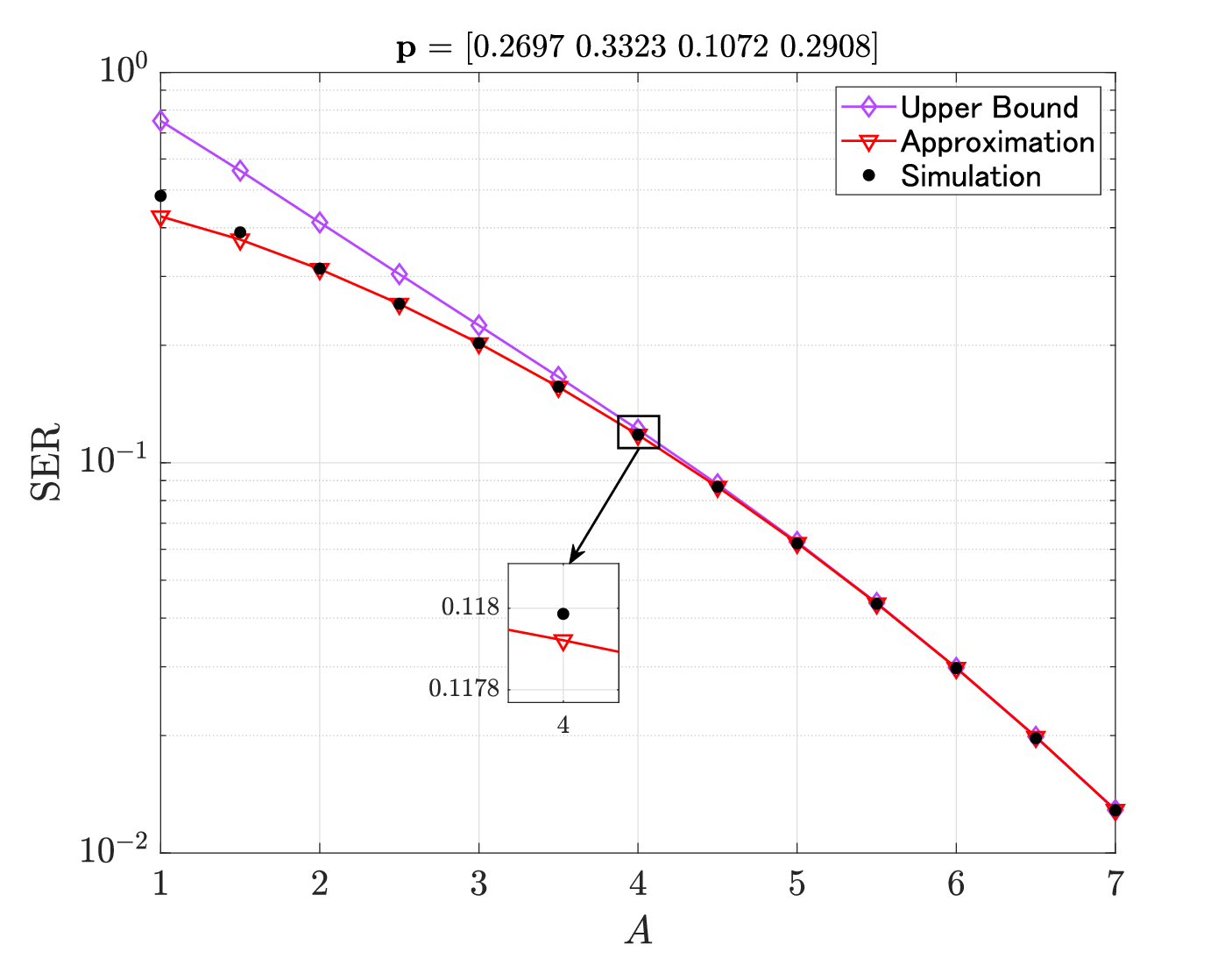}
         \caption{4-PAM.}
         \label{4PAM_1}
     \end{subfigure}
     \begin{subfigure}[b]{0.48\textwidth}
         \centering
         \includegraphics[width=\textwidth, height = 5.2cm]{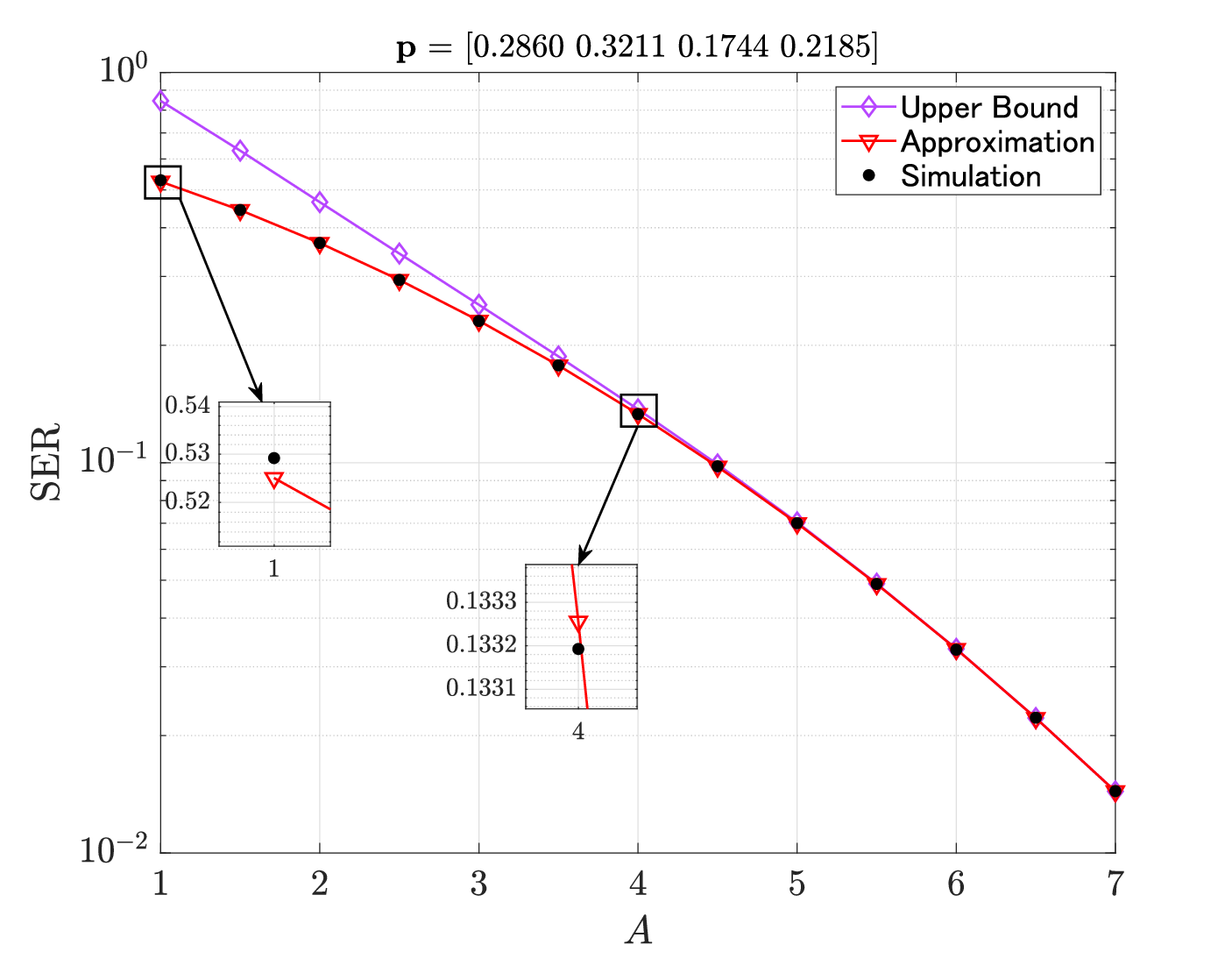}
         \caption{4-PAM.}
         \label{4PAM_2}
     \end{subfigure}
     \begin{subfigure}[b]{0.48\textwidth}
         \centering
         \includegraphics[width=\textwidth, height = 5.2cm]{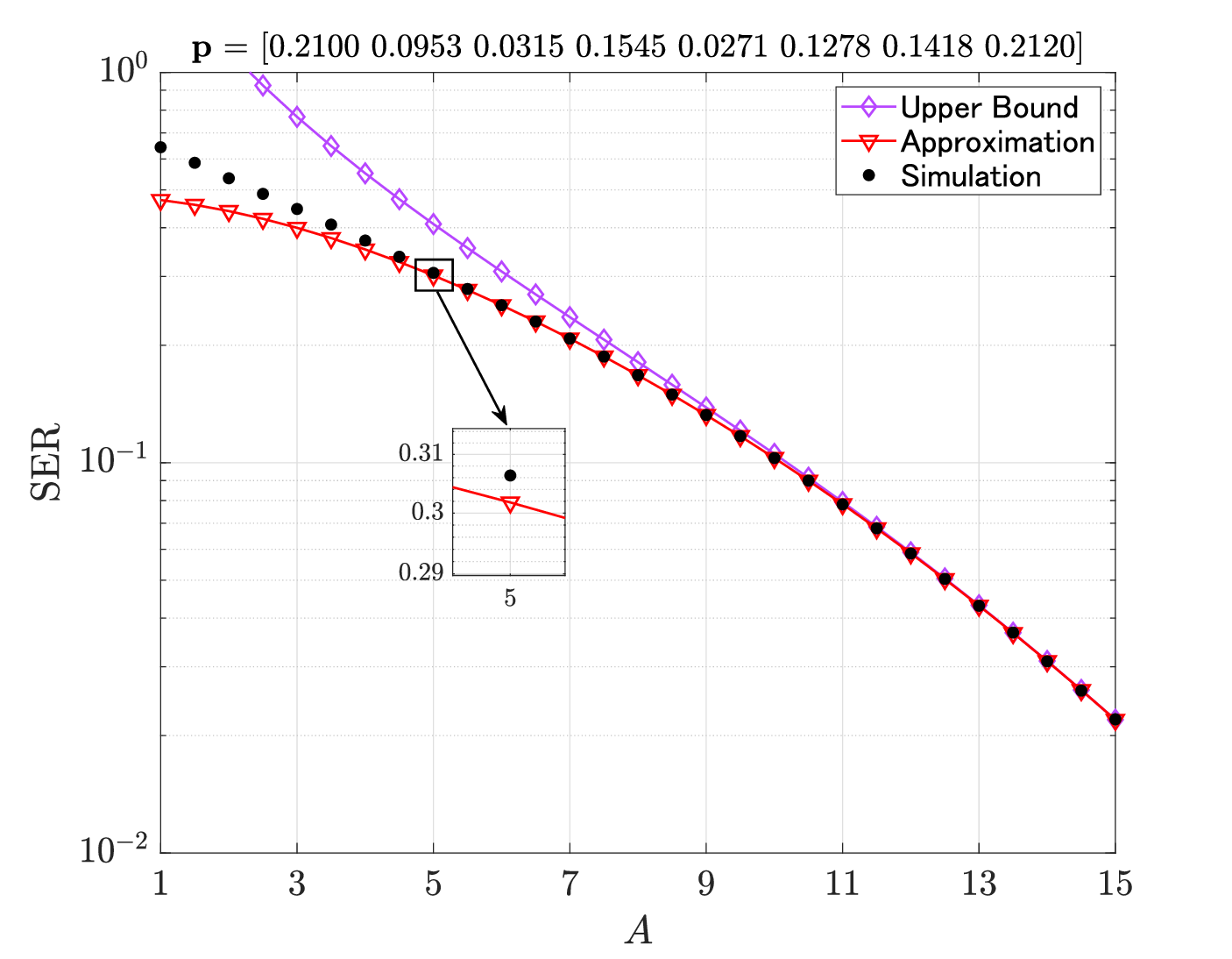}
         \caption{8-PAM.}
         \label{8PAM_1}
     \end{subfigure}
     \begin{subfigure}[b]{0.48\textwidth}
         \centering
         \includegraphics[width=\textwidth, height = 5.2cm]{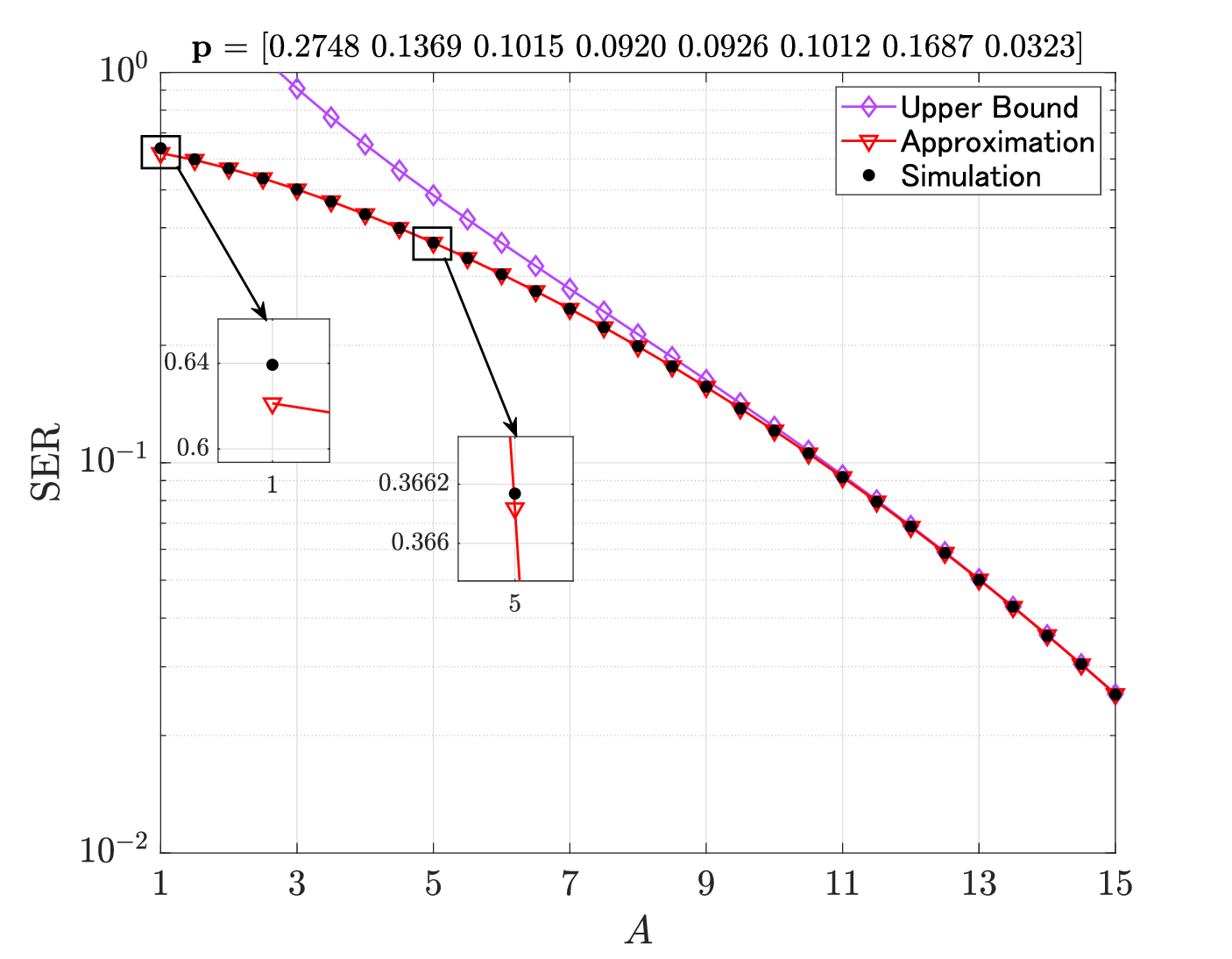}
         \caption{8-PAM.}
         \label{8PAM_2}
     \end{subfigure}
        \caption{Upper bound and approximate SER of PCS-PAM.}
        \label{PCS-SER}
\end{figure*}
\subsubsection{Numerial Examples}
We evaluate the tightness of the derived upper bound and approximate SER with respect to the exact value in Figs.~\ref{4PAM_1},~\ref{4PAM_2}, ~\ref{8PAM_1}, and \ref{8PAM_2} for randomly generated constellation distributions. $4$-PAM and $8$-PAM constellations are considered for comparisons. For convenience, $h_{\text{U}}\gamma\eta$ and $\sigma_{\text{U}}$ are both normalized to unity, i.e., $h_{\text{U}}\gamma\eta = 1$ and  $\sigma_{\text{U}} = 1$. The upper bound and approximate SER expressions are thus functions of $A$ and $\mathbf{p}$. Monte-Carlo simulations of the MAP detector in \eqref{MAP-detector} are also provided as the exact SER value. It is observed that while the upper bound and approximate expressions are tight at the high-value region of $A$ (i.e., high SNR region), they generally diverge from the exact value as $A$ decreases. Overall, the approximation provides a better fit to the exact value at the low value of $A$ as the upper bound tends to be very loose at that regime. It is also evident that the tightness of the approximation is dependent on the average discrepancy of adjacent symbols' transmission probabilities.  For example, the average discrepancy of adjacent symbols' transmission probabilities in Figs.~\ref{4PAM_1}, \ref{4PAM_2}, \ref{8PAM_1}, and \ref{8PAM_2} are 0.1571, 0.0753, 0.0877, and 0.0566, respectively, and it can be seen that the approximations in Figs.~\ref{4PAM_2} and \ref{8PAM_2} fit better with the exact value than those in Figs.~\ref{4PAM_1} and \ref{8PAM_1}.


\vspace{-1.5mm}
\subsection{Flickering Mitigation}
Recall that the average emitted optical power under arbitrary symbol probabilities induced by PCS is $\mathbb{E}[P_t] = \eta (I_{\text{DC}} + \mathbf{a}^T\mathbf{p})$. 
To minimize the effect of flickering resulting from the variation of $\mathbf{a}^T\mathbf{p}$, a straightforward way is to limit the magnitude of this variation to be sufficiently small, such as 
\begin{align}
    \left|\mathbf{a}^T\mathbf{p}\right| \leq \alpha I_{\text{DC}},
    \label{flickering-constraint}
\end{align}
where $\alpha$ is a chosen constant. By choosing $\alpha$ small enough, $\mathbb{E}[P_t]$ can be made approximately $\eta I_{\text{DC}}$, which is the average emitted optical power in the case of uniform signaling. Obviously, flickering is completely removed when $\alpha = 0$ is chosen. This could, however, result in a stringent constraint on $\mathbf{p}$, leading to a reduced secrecy performance. Hence, a sufficiently small $\alpha$, which may cause negligible flickering, could still be acceptable.   
\vspace{-1.5mm}
\subsection{Symmetric Constellation Distribution}
To ensure the symmetry of the constellation distribution in the case that the PAS architecture is employed, the following constraint is imposed on $\mathbf{p}$
\begin{align}
    \mathbf{S}\mathbf{p} = \mathbf{0}_{\frac{M}{2}},
    \label{symetry-constraint}
\end{align}
where the matrix $\mathbf{S} \in \mathbb{R}^{\frac{M}{2} \times M}$ has the form
\begin{align}
    \mathbf{S} = \begin{bmatrix}
        1 & 0 & 0 & \cdots & 0 & -1\\
        0 & 1 & 0 & \cdots & -1 & 0 \\
        & & & \vdots & & \\
        0 & \cdots & 1 & -1 & \cdots & 0 \\
    \end{bmatrix}.
\end{align}
That is, $\mathbf{S}[i,~i] = 1$, $\mathbf{S}[i,~M-i+1] = -1$ $\forall i = 1, 2,..., \frac{M}{2}$, $\mathbf{S}[i,~j] = 0$ $\forall j \neq i$ and $j \neq M-i+1$. It should be noted that the symmetry of the constellation distribution results in $\mathbf{a}^T\mathbf{p} = 0$ (since $\mathbf{a}$ is symmetric around zero), which implies the satisfaction of the flickering mitigation constraint. Thus, either constraint in \eqref{flickering-constraint} or \eqref{symetry-constraint} is considered in our PCS design. 
\section{PCS Design and Solution}
\subsection{Known Eve's CSI}
The objective of our PCS design is to maximize the secrecy capacity while satisfying the reliability of Bob's channel, flickering mitigation, and symmetric constellation distribution constraint. To ensure that the BER of Bob's channel is below the pre-FEC BER threshold, the upper bound BER expression in \eqref{PCS-BER} is employed, which results in the following design problem\footnote{Regarding the formulation of problem \eqref{OptProb2}, when a symmetric constellation distribution is required (e.g., due to the use of the PAS architecture), constraint in \eqref{constraint22} is replaced by \eqref{symetry-constraint}. Note that since both \eqref{constraint22} and \eqref{symetry-constraint} are linear constraints, using \eqref{symetry-constraint} does not require a different solving approach. }
\begin{subequations}
\label{OptProb2}
    \begin{alignat}{2}
        &\underset{\mathbf{p}}{\text{maximize}} & \hspace{2mm} & C_{\text{s}}(\mathbf{p}) \label{obj2}\\
        &\text{subject to }  &  & \nonumber \\
        & & & P^{\text{B}}_{\text{e,ub}}(\mathbf{p}) \leq P_{\text{e}}^{\text{pre}},
        \label{constraint21} \\ 
        & & & \left|\mathbf{a}^T\mathbf{p}\right| \leq \alpha I_{\text{DC}}, \label{constraint22} \\
         & & & \eqref{constraint11},~ \eqref{constraint12} \nonumber.
    \end{alignat}
\end{subequations}
It can be seen that while \eqref{constraint22}, \eqref{constraint11}, and \eqref{constraint12} are simple linear constraints, \eqref{constraint21} is a complex non-convex one due to the following proposition. 
\begin{Proposition}
$P^{\rm{U}}_{\rm{e, ub}}(\mathbf{p})$ is a concave function of $\mathbf{p}$.
\end{Proposition}

\begin{proof}
    The proof is given in Appendix B. 
\end{proof}
\noindent Due to \textbf{Proposition 2}, \eqref{constraint21} is not convex. As a result, \eqref{OptProb2} is a non-convex optimization problem, which generally renders solving the global solution very challenging. 

In this paper, the CCCP is employed to solve a sub-optimal solution \eqref{OptProb2}, which often exhibits a lower computational complexity. The principle idea of CCCP is to approximate the non-convex objective function and/or non-convex constraints to their corresponding ones. Then, a local optimum can be found via an iterative procedure. Specifically for our problem, due to its concavity, the BER expression in \eqref{PCS-BER} is upper bounded by its first-order linear expansion as 
\begin{align}
    P^{\text{B}}_{\text{e, ub}}(\mathbf{p}) \leq P^{\text{B}}_{\text{e, ub}}\left(\mathbf{p}^{(k)}\right) + \nabla^T P^{\text{B}}_{\text{e, ub}}(\mathbf{p})\left(\mathbf{p} - \mathbf{p}^{(k)}\right),
    \label{first-order-approximation}
\end{align}
where $\mathbf{p}^{(k)}$ is the solution to \eqref{OptProb2} at the $k$-th iteration. The $m$-th element of the gradient $\nabla P^{\text{B}}_{\text{e, ub}}(\mathbf{p})$ is in \eqref{gradient-BER}, which is on top of the next page. 
\begin{figure*}
\begin{align}
     \frac{\partial P^{\text{B}}_{\text{e, ub}}(\mathbf{p})}{\partial p_m}  & = \frac{1}{2\log_2M}\left(\frac{\partial\sum_{\substack{n = 1 \\ n\neq m}}^Mp_m\text{erfc}\left(\frac{2\sigma^2_{\rm{B}}\log\frac{p_m}{p_n} + d^2_{{\rm{B}}, m, n}}{2\sqrt{2}\sigma_{{\rm{B}}}|d_{{\rm{B}}, m, n}|}\right)}{\partial p_m} \right. + \left.\frac{\partial\sum_{\substack{n = 1 \\ n\neq m}}^Mp_n\text{erfc}\left(\frac{2\sigma^2_{\rm{B}}\log\frac{p_n}{p_m} + d^2_{{\rm{B}}, n, m}}{2\sqrt{2}\sigma_{{\rm{B}}}|d_{{\rm{B}}, n, m}|}\right)}{\partial p_m}\right) \nonumber \\
    & = \frac{1}{\log_2M}\left(\frac{1}{2}\sum_{\substack{n = 1 \\ n\neq m}}^M\text{erfc}\left(\frac{2\sigma^2_{\rm{B}}\log\frac{p_m}{p_n} + d^2_{{\rm{B}}, m, n}}{2\sqrt{2}\sigma_{{\rm{B}}}|d_{{\rm{B}}, m, n}|}\right) \right.
     \left.- \frac{\sigma_{\text{B}}}{\sqrt{2\pi}}\sum_{\substack{n = 1 \\ n\neq m}}^M\frac{1}{|d_{\text{B}, m, n}|}\text{exp}\left(-\left(\frac{2\sigma^2_{\rm{B}}\log\frac{p_m}{p_n} + d^2_{{\rm{B}}, m, n}}{2\sqrt{2}\sigma_{{\rm{B}}}|d_{{\rm{B}}, m, n}|}\right)^2\right) \right.\nonumber \\
    & \hspace{18mm} \left. +\frac{\sigma_{\text{B}}}{\sqrt{2\pi}}\sum_{\substack{n = 1 \\ n\neq m}}^M\frac{p_n}{p_m|d_{\text{B}, n, m}|}\text{exp}\left(-\left(\frac{2\sigma^2_{\rm{B}}\log\frac{p_n}{p_m} + d^2_{{\rm{B}}, n, m}}{2\sqrt{2}\sigma_{{\rm{B}}}|d_{{\rm{B}}, n, m}|}\right)^2\right)\right).
    \label{gradient-BER}
\end{align}
\rule{\textwidth}{0.4pt}
\end{figure*}
Replacing $ P^{\text{B}}_{\text{e, ub}}(\mathbf{p})$ by its upper bound in \eqref{first-order-approximation}, the CCP solve \eqref{OptProb2} by iteratively solving the following surrogate problem 
\begin{subequations}
\label{OptProb3}
    \begin{alignat}{2}
        &\underset{\mathbf{p}}{\text{maximize}} & \hspace{2mm} & C_{\text{s}}(\mathbf{p}) \label{obj3}\\
        &\text{subject to }  &  & \nonumber \\
        & & & P^{\text{B}}_{\text{e, ub}}\left(\mathbf{p}^{(k)}\right) + \nabla^T P^{\text{B}}_{\text{e, ub}}(\mathbf{p})\left(\mathbf{p} - \mathbf{p}^{(k)}\right) \leq P_{\text{e}}^{\text{pre}},
        \label{constraint31} \\ 
        & & & \eqref{constraint22},~\eqref{constraint11},~\eqref{constraint12}, \nonumber
    \end{alignat}
\end{subequations}
which is a convex optimization problem and thus can be efficiently solved using interior-point methods. 
A solving algorithm using CCCP is summarized as follows. 
\begin{algorithm2e}[http]
\SetAlgoLined 
\caption{CCCP algorithm to solve \eqref{OptProb2}}
\label{alg.1}
Set $L_{\text{max}}$ and  $\epsilon > 0$ as the maximum number of iterations and the error tolerance. \\
Choose a  starting point ${\mathbf{p}}^{(0)}$ for \eqref{OptProb2} and 
set $k \leftarrow 1$. \\
\While{convergence = \textbf{False} and $k \leq L_{\text{max}}$}{
Solve \eqref{OptProb3} using $\mathbf{p}^{(k-1)}$  from the ($k$-1)-th iteration.\\
\eIf{$\left|\frac{ C_{\text{s}}\left(\mathbf{p}^{(k)}\right) - C_{\text{s}}\left({\mathbf{p}}^{(k-1)}\right)}{ C_{\text{s}}\left({\mathbf{p}}^{(k-1)}\right)}\right| \leq \epsilon$ }{
convergence = \textbf{True}. \\
$\mathbf{p}^{*} \leftarrow {\mathbf{p}}^{(k)}$. \\
}
{convergence $\leftarrow$ \textbf{False}.\\}
$k \leftarrow k + 1$. \\
}
Return the solution $\mathbf{p}^{*}$.
\end{algorithm2e} 
\subsection{Unknown Eve's CSI}
In most practical scenarios, Eve is a passive unauthorized user, who does not actively exchange information with the transmitter. In such cases, the CSI of Eve's channel may not be available by the transmitter; thus, it is not valid to use the secrecy capacity formula in \eqref{secrecy_capacity} for the PCS design as in the case of known Eve's CSI. 
When the instantaneous CSI of Eve's channel is unknown, one can consider the average secrecy capacity with respect to Eve's channel, which can be given by
\begin{align}
  \overline{C}_s(\mathbf{p}) =  & -\int_{-\infty}^\infty p(\overline{y}_{\text{B}})\log_2 p(\overline{y}_{\text{B}})\text{d}\overline{y}_{{\text{B}}} - \frac{1}{2}\log_2\left(2\pi e\sigma^2_{\text{B}}\right)\nonumber \\ 
   & + \mathbb{E}_{h_{\text{E}}}\left[\int_{-\infty}^\infty p(\overline{y}_{\text{E}})\log_2 p(\overline{y}_{\text{E}})\text{d}\overline{y}_{{\text{E}}}  + \frac{1}{2}\log_2\left(2\pi e\sigma^2_{\text{E}}\right)\right].
   \label{avg_secrecy_capacity}
\end{align}
An issue of the above expression is the average of the differential entropy of $\overline{y}_{\text{E}}$, which does not lead to a tractable form. Therefore, we consider an alternative approach using an upper bound on the capacity of Eve's channel. Recall from \eqref{mutual_information} that the capacity of Eve's channel is given by
\begin{align}
    C_{\text{E}} = \mathbb{H}(\overline{y}_{\text{E}}) - \mathbb{H}(n_{\text{E}}) = \mathbb{H}(h_{\text{E}}\gamma\eta s + n_{\text{E}}) - \frac{1}{2}\log_2\left(2\pi e\sigma^2_{\text{E}}\right).
\end{align}
Let us consider the random variable $z = h_{\text{E}}\gamma\eta s$ whose expected value is $\mu_z = h_{\text{E}}\gamma\eta\mathbf{a}^T\mathbf{p}$. Since $z$ is bounded between $-h_{\text{E}}\gamma\eta A$ and $h_{\text{E}}\gamma\eta A$, according to the Bhatia–Davis inequality \cite{Rajendra2000}, the variance of $z$ is bounded by $\sigma^2_{z} \leq \left(h_{\text{E}}\gamma\eta\right)^2\left(A^2 - \mathbf{a}^T\mathbf{p}\mathbf{p}^T\mathbf{a}\right)$. Following a well-known result that $\mathbb{H}(x) \leq \frac{1}{2}\log_2\left(2\pi e \sigma^2_x\right)$ (here, $\sigma^2_x$ is the variance of the random variable $x$), an upper bound on $C_{\text{E}}$ is thus
\begin{align}
    C_{\text{E}} \leq \frac{1}{2}\log_2\left(1 + \frac{\left(h_{\text{E}}\gamma\eta\right)^2\left(A^2 - \mathbf{a}^T\mathbf{p}\mathbf{p}^T\mathbf{a}\right)}{\sigma^2_{\text{E}}}\right),
\end{align}
which results in an average lower bound on the secrecy capacity as 
\begin{align}
    \overline{C}^L_s(\mathbf{p}) = & -\int_{-\infty}^\infty p(\overline{y}_{\text{B}})\log_2 p(\overline{y}_{\text{B}})\text{d}\overline{y}_{{\text{B}}} - \frac{1}{2}\log_2\left(2\pi e\sigma^2_{\text{B}}\right)\nonumber \\ 
   & -\frac{1}{2}\mathbb{E}_{h_{\text{E}}}\left[\log_2\left(1 + \frac{\left(h_{\text{E}}\gamma\eta\right)^2\left(A^2 - \mathbf{a}^T\mathbf{p}\mathbf{p}^T\mathbf{a}\right)}{\sigma^2_{\text{E}}}\right)\right].
   \label{lower-bound}
\end{align}
The above expression is, however, still difficult to handle due to the intractability of the average term. Here, we present the following tractable estimation for \eqref{lower-bound} as
\begin{align}
    \widetilde{C}^L_s(\mathbf{p}) = & -\int_{-\infty}^\infty p(\overline{y}_{\text{B}})\log_2 p(\overline{y}_{\text{B}})\text{d}\overline{y}_{{\text{B}}} - \frac{1}{2}\log_2\left(2\pi e\sigma^2_{\text{B}}\right)\nonumber \\ 
   & -\frac{1}{2}\log_2\left(1 + \frac{\left(\overline{h}_{\text{E}}\gamma\eta\right)^2\left(A^2 - \mathbf{a}^T\mathbf{p}\mathbf{p}^T\mathbf{a}\right)}{\overline{\sigma}^2_{\text{E}}}\right),
   \label{lower-bound-estimation}
\end{align}
with $\overline{h}_{\text{E}}$ being the average gain of Eve's channel and $\overline{\sigma}^2_{\text{E}}$ being the average noise power, which is given by
\begin{align}
\overline{\sigma}^2_\text{E} = B_{\text{mod}}\left(2e\gamma \overline{h}_{\text{E}}P_{\text{t}} + 4\pi e \gamma A_r\chi_{\text{amb}}\left(1-\cos\Psi_{\text{E}}\right) + i^2_{\text{amp}}\right).
\end{align}

To calculate $\overline{h}_{\text{E}}$, we assume that Eve's position is uniformly distributed within the receiver plane. Denoting $r_{\text{E}}$ as the distance from Eve to the center of the beam footprint. According to \eqref{channel_gain}, the gain of Eve's channel is $0$ when $\psi_{\text{E}} \geq \Psi_{\text{E}}$, or equivalently $\cos\left(\psi_{\text{E}}\right) = \frac{L}{\sqrt{L^2 + r^2_{\text{E}}}} \leq \cos\left(\Psi_{\text{E}}\right)$, which leads to $r_{\text{E}} \geq L\tan\left(\Psi_{\text{E}}\right)$. To calculate the average gain of Eve's channel, one can reformulate \eqref{channel_gain} to 
\begin{align}
    h_\text{E} =
  \begin{cases}
    \Xi(\psi_{\text{E}})L^{l+1}\left(L^2 + r^2_{\text{E}}\right)^{-\frac{l+3}{2}}     & \quad 0 \leq \psi_{\text{E}} \leq \Psi_{\text{E}},\\
    0  & \quad \psi_{\text{E}} > \Psi_{\text{E}}.
  \end{cases}
  \label{channel_gain_reformulation}
\end{align}
where $\Xi(\psi_{\text{E}}) = \frac{A_{\text{r}}(l+1)}{2\pi}T(\psi_{\text{E}})g(\psi_{\text{E}})$.
The average gain of Eve's channel is thus given by 
\begin{align}
    \overline{h}_{\text{E}} = \frac {\Xi(\psi_{\text{E}})L^{l}}{\tan\left(\Psi_{\text{E}}\right)}\int_{0}^{L\tan\left(\Psi_{\text{E}}\right)}\left(L^2 + r^2_{\text{E}}\right)^{-\frac{l+3}{2}}dr_{\text{E}}.
\end{align}
The above integration generally does not lead to a simple closed-form solution. Instead, it can be evaluated in terms of the hypergeometric function ${}_pF_q(\cdot;\cdot;\cdot)$\cite{Abramowitz:1964:HMF} as 
\begin{align}
    \overline{h}_{\text{E}} =\Xi(\psi_{\text{E}})L^{-2}~  {}_2F_1\left(\frac{1}{2}, \frac{l+3}{2}; \frac{3}{2}; -\tan^2\left(\Psi_{\text{E}}\right)\right).
    \label{avg_hE_hypergeometric}
\end{align}
Note that for a special value of $l = 1$ (corresponding to the LED's semiangle at half power $\Phi_{1/2} = 60^\circ$), the expression in \eqref{avg_hE_hypergeometric} reduces to the following simple form
\begin{align}
    \overline{h}_{\text{E}} = \frac{\Xi\left(\psi_{\text{E}}\right)\left(\sin\left(2\Psi_{\text{E}}\right) + 2\Psi_{\text{E}}\right)}{4L^2\tan\left(\Psi_{\text{E}}\right)}.
\end{align}

Using \eqref{lower-bound-estimation}, the PCS design problem is formulated as 
\begin{subequations}
\label{OptProb5}
    \begin{alignat}{2}
        &\underset{\mathbf{p}}{\text{maximize}} & \hspace{2mm} & \widetilde{C}^L_s(\mathbf{p}) \label{obj5}\\
        &\text{subject to }  &  & \nonumber \\
         & & & \eqref{constraint21},~\eqref{constraint22},~ \eqref{constraint11},~ \eqref{constraint12} \nonumber.
    \end{alignat}
\end{subequations}
In contrast to objective functions of problems \eqref{OptProb2}, the objective function in \eqref{obj5} appears to be nonconcave. To handle the above problem, we first introduce the variable transformation $t = \mathbf{a}^T\mathbf{p}\mathbf{p}^T\mathbf{a}$ and rewrite \eqref{lower-bound-estimation} to
\begin{align}
    \widetilde{C}^L_s(\mathbf{p}, t) = & -\int_{-\infty}^\infty p(\overline{y}_{\text{B}})\log_2 p(\overline{y}_{\text{B}})\text{d}\overline{y}_{{\text{B}}} - \frac{1}{2}\log_2\left(2\pi e\sigma^2_{\text{B}}\right)\nonumber \\ 
   & -\frac{1}{2}\log_2\left(1 + \frac{\left(\overline{h}_{\text{E}}\gamma\eta\right)^2\left(A^2 - t\right)}{\overline{\sigma}^2_{\text{E}}}\right).
   \label{lower-bound-estimation_1}
\end{align}
The problem in \eqref{OptProb5} is then reformulated as 
\begin{subequations}
\label{OptProb6}
    \begin{alignat}{2}
        &\underset{\mathbf{p},~t}{\text{maximize}} & \hspace{2mm} & \widetilde{C}^L_s(\mathbf{p}, t) \label{obj6}\\
        &\text{subject to }  &  & \nonumber \\
        & & & t \leq \mathbf{a}^T\mathbf{p}\mathbf{p}^T\mathbf{a}, \label{constraint61} \\
         & & & \eqref{constraint21},~\eqref{constraint22},~ \eqref{constraint11},~ \eqref{constraint12} \nonumber.
    \end{alignat}
\end{subequations}
We note that the objective function in \eqref{obj6} is not concave and the constraint in \eqref{constraint61} is not convex. The CCCP, therefore, is utilized to solve a local solution to \eqref{OptProb6}. Specifically, the objective function and the right-hand side of \eqref{constraint61} are replaced by their respective linear lower bounds 
\begin{align}
    &\widetilde{C}^L_s(\mathbf{p}, t) \geq \widetilde{C}^L_s(\mathbf{p}, t, t^{(k)}) =   -\int_{-\infty}^\infty p(\overline{y}_{\text{B}})\log_2 p(\overline{y}_{\text{B}})\text{d}\overline{y}_{{\text{B}}} \nonumber \\ &  - \frac{1}{2}\log_2\left(2\pi e\sigma^2_{\text{B}}\right)
    -\frac{1}{2}\log_2\left(1 + \frac{\left(\overline{h}_{\text{E}}\gamma\eta\right)^2\left(A^2 - t^{(k)}\right)}{\overline{\sigma}^2_{\text{E}}}\right) \nonumber \\
    & + \frac{1}{2\log(2)\left(1 + \frac{\left(\overline{h}_{\text{E}}\gamma\eta\right)^2\left(A^2 - t^{(k)}\right)}{\overline{\sigma}^2_{\text{E}}}\right)}\left(t - t^{(k)}\right),
\end{align}
and
\begin{align} \mathbf{a}^T\mathbf{p}^{(k)}\left[\mathbf{p}^{(k)}\right]^T\mathbf{a} + 2\mathbf{a}^T\mathbf{p}^{(k)}\mathbf{a}^T\left(\mathbf{p} - 
    \mathbf{p}^{(k)}\right) \leq \mathbf{a}^T\mathbf{p}\mathbf{p}^T\mathbf{p},
\end{align}
where $t^{(k)}$ and $\mathbf{p}^{(k)}$ are the solutions to \eqref{OptProb6} at the $k$-th iteration of the CCCP. A similar CCCP algorithm to that described in \textbf{Algorithm 1} is then used to find a local maximum to \eqref{OptProb6}.

Now, we examine the design considering the constraint of symmetric constellation distribution. Notice that once the constraint is imposed, the lower-bound secrecy capacity in \eqref{lower-bound} becomes 
\begin{align}
   \overline{C}^L_s(\mathbf{p}) = & -\int_{-\infty}^\infty p(\overline{y}_{\text{B}})\log_2 p(\overline{y}_{\text{B}})\text{d}\overline{y}_{{\text{B}}} - \frac{1}{2}\log_2\left(2\pi e\sigma^2_{\text{B}}\right)\nonumber \\ 
   & -\frac{1}{2}\mathbb{E}_{h_{\text{E}}}\left[\log_2\left(1 + \frac{\left(h_{\text{E}}\gamma\eta A\right)^2}{\sigma^2_{\text{E}}}\right)\right],
   \label{lower-bound-1} 
\end{align}
which is a concave function of $\mathbf{p}$ due to the concavity of the differential entropy of $\overline{y}_{\text{B}}$ \cite[\textbf{Theorem 2.7.3}]{Cover2006}. Since the second and third terms of $\eqref{lower-bound-1}$ are constant, the optimal PCS design problem can be formulated as 
\begin{subequations}
\label{OptProb7}
    \begin{alignat}{2}
        &\underset{\mathbf{p}}{\text{minimize}} & \hspace{2mm} & \int_{-\infty}^\infty p(\overline{y}_{\text{B}})\log_2 p(\overline{y}_{\text{B}})\text{d}\overline{y}_{{\text{B}}} \label{obj7}\\
        &\text{subject to }  &  & \nonumber \\
         & & & \eqref{symetry-constraint},~\eqref{constraint21},~ \eqref{constraint11},~ \eqref{constraint12} \nonumber,
    \end{alignat}
\end{subequations}
which is a convex optimization problem. 
\vspace{-2mm}
\subsection{QoS-based PCS Design}
Aside from maximizing the secrecy capacity, physical layer security designs from the QoS-based perspective have also been addressed in literature \cite{Liao2010, Valeria2024} where QoS in terms of the BER (or equivalently the signal-to-noise ratio) has been considered. In this work, we examine a QoS-based PCS design that aims to maximize the BER of Eve's channel.  Using the approximate SER in \eqref{approximate-SER} as a reasonable evaluation for the SER of Eve's channel, the problem is formulated as   
\begin{subequations}
\label{OptProb8}
    \begin{alignat}{2}
        &\underset{\mathbf{p}}{\text{maximize}} & \hspace{2mm} & P^{\text{E}}_{\text{e, ap}}(\mathbf{p}) \label{obj8}\\
        &\text{subject to }  &  & \nonumber \\
        & & & \eqref{constraint21},~\eqref{constraint22}, ~\eqref{constraint11},~ \eqref{constraint12} \nonumber,
    \end{alignat}
\end{subequations}
where $P^{\text{E}}_{\text{e, ap}}(\mathbf{p}) \approx \frac{1}{\log_2(M)}P^{\text{E}}_{\text{s, ap}}(\mathbf{p})$.
\begin{Proposition}
$P^{\rm{U}}_{\rm{e, ap}}(\mathbf{p})$ is a concave function of $\mathbf{p}$.
\end{Proposition}
\noindent The proof is similar to that of \textbf{Proposition 2}, thus is omitted for brevity. Due to the concavity of $P^{\rm{E}}_{\rm{e, ap}}(\mathbf{p})$, we can handle \eqref{OptProb8} using the same approach to solving \eqref{OptProb2}. 
\section{Simulation Results and Discussions}
This section presents the secrecy capacity and BER performance of the proposed PCS designs in comparison with the uniform PAM signaling.
It is assumed that $I_{\text{max}} \gg I_{\text{DC}}$ and $I_{\text{min}} = 0$, resulting in $\text{min}\left(I_{\text{max}} - I_{\text{DC}}, I_{\text{DC}} - I_{\text{min}}\right) = I_{\text{DC}}$. 
We then denote $P_{t, \text{DC}} = \eta I_{\text{DC}}$, which is the average emitted optical power (i.e., $\mathbb{E}[P_t]$) in the case of uniform symbol probability. Note that in the case of PCS-PAM, $P_{t, \text{DC}}$ can also be well approximated by $\mathbb{E}[P_t]$ since a very small $\alpha$ (i.e., $\alpha = 0.01$ as in our simulations) should be chosen to minimize the flickering effect. Other simulation parameters are given as follows. 
Height of the LED $L= 3$ m, LED semi-angle at half power $\Phi_{1/2} = 60^\circ$, LED conversion factor $\eta = 0.44$ W/A, PD active area, $A_r =  1~\text{cm}^2$, PD responsivity $\gamma = 0.54 $ A/W,  PD field of view (FoV) $\Psi_{\text{U}} = 70^\circ$, optical filter gain $T_s(\psi_{\text{U}}) =$1, refractive index of the concentrator $\kappa = 1.5$, modulation bandwidth $B = 20$ MHz, ambient light photocurrent $\chi_{\text{amb}} = 10.93$ $\text{A}/(\text{m}^2\cdot\text{Sr})$, pre-amplifier noise density current $i_{\text{amb}} = 5~\text{pA/Hz}^{-1/2}$. 

\begin{figure}[ht]
     \centering
     \includegraphics[height = 6cm, width = 0.48\textwidth]{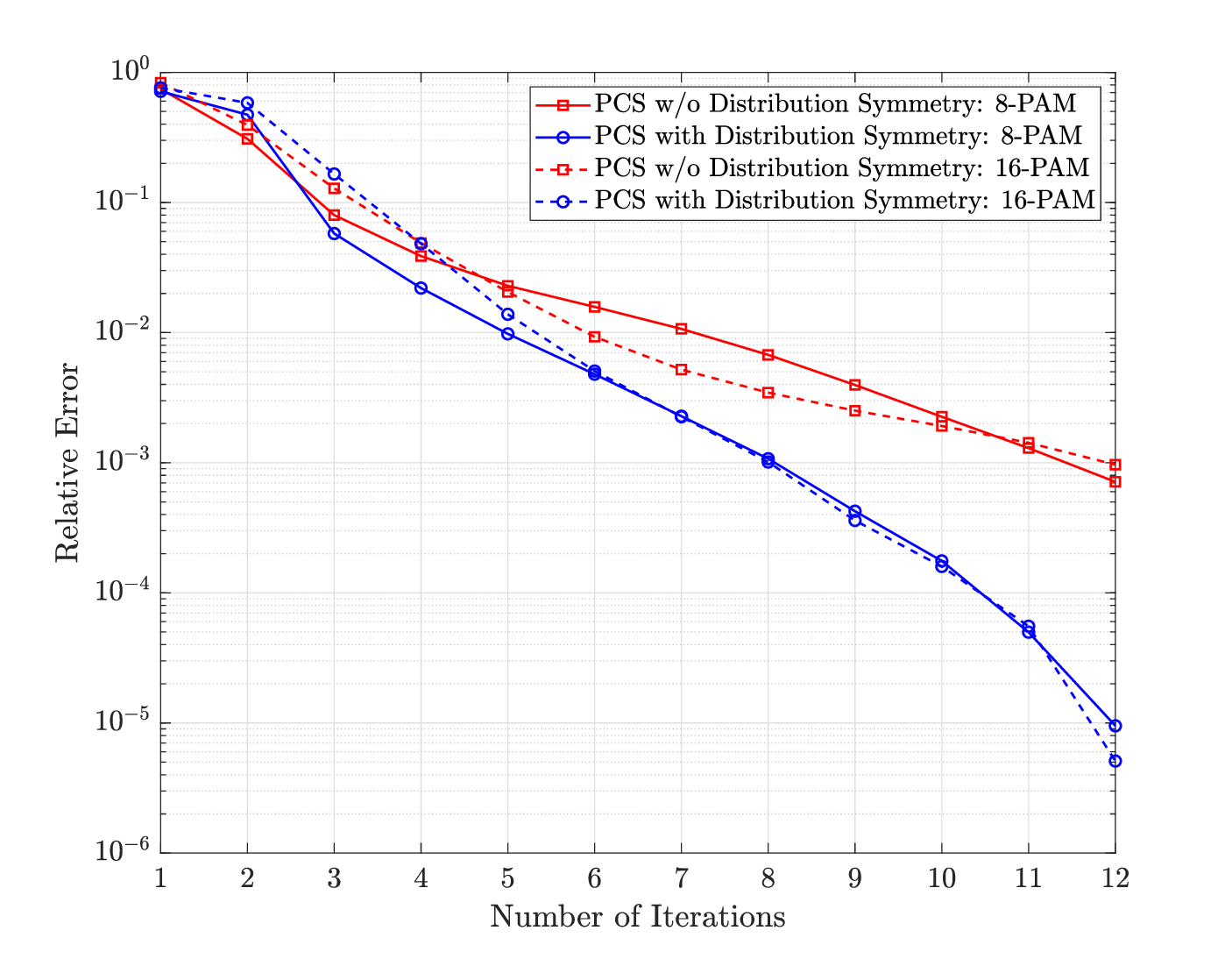}
     \caption{Convergence behavior of \textbf{Algorithm 1}.}
     \label{Convergence}
\end{figure}
Firstly, we show in Fig.~\ref{Convergence} the convergence behavior of \textbf{Algorithm 1} with $P_{t, \text{DC}} = 25$ dBm being set. Here, the relative error of the objective value is displayed in accordance with the number of iterations of the algorithm for the PCS designs of 8-PAM and 16-PAM constellations with and without the constraint of distribution symmetry. Again, we emphasize that in the PCS design without distribution symmetry constraint, the flickering mitigation constraint is employed, while in the design with distribution symmetry constraint, the flickering mitigation constraint is not used. Note that since CCCP is a heuristic method, the obtained local solution may depend on the starting point. Therefore, we run the algorithm with 10000 different starting points to obtain the average results. At the target error tolerance  $\epsilon = 10^{-2}$, less than 7 iterations are required on average. It is also observed that the PCS designs with the distribution symmetry constraint (which is more stringent than the flickering mitigation one) generally converge quicker than those without it. 

\begin{figure*}
     \centering
     \begin{subfigure}[b]{0.49\textwidth}
         \centering
         \includegraphics[height = 5.8cm, width = \textwidth]{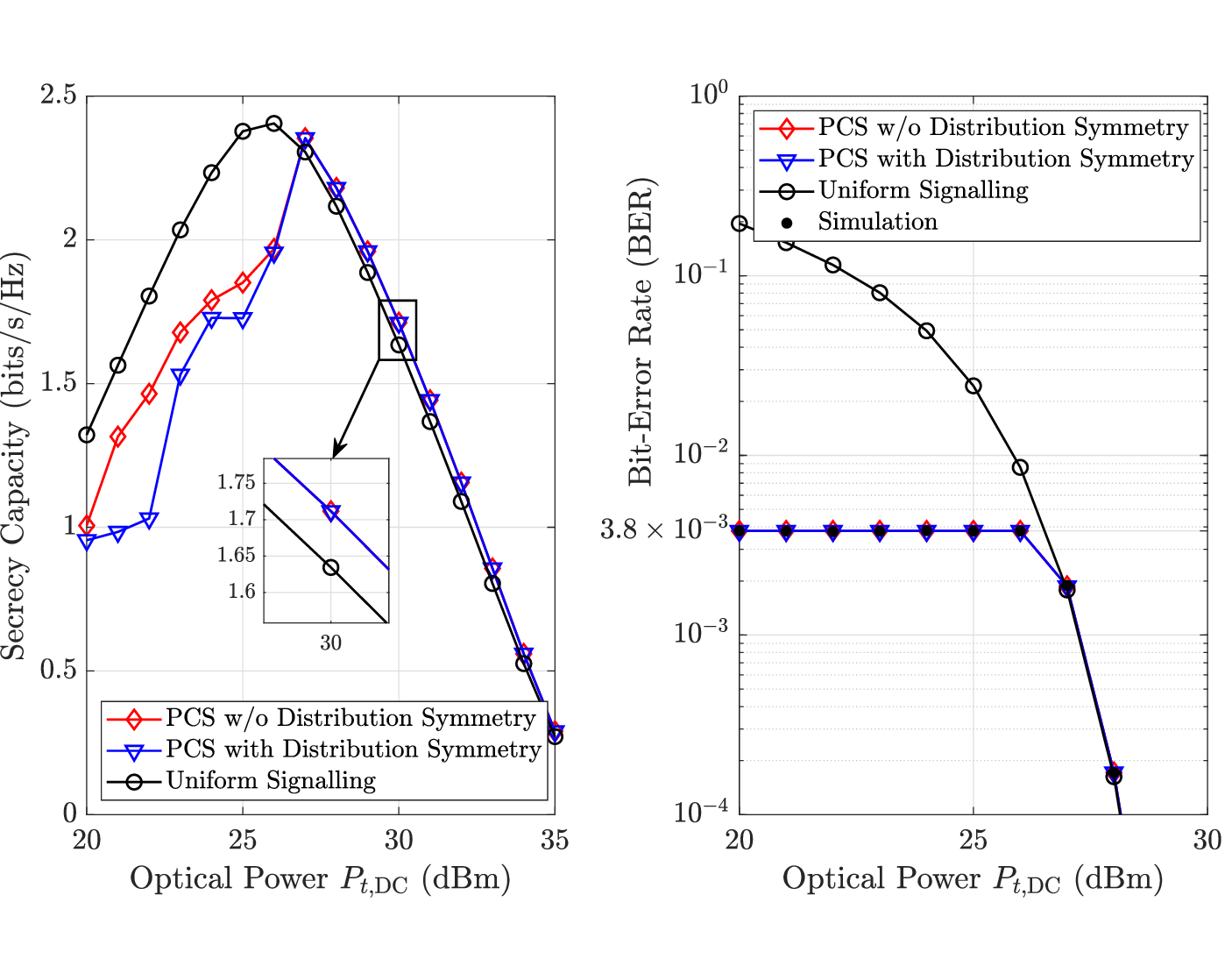}
         \caption{8-PAM.}
         \label{8PAM_Comparison_KnownEve}
     \end{subfigure}
     \hfill
     \begin{subfigure}[b]{0.49\textwidth}
         \centering
         \includegraphics[height = 5.8cm, width = \textwidth]{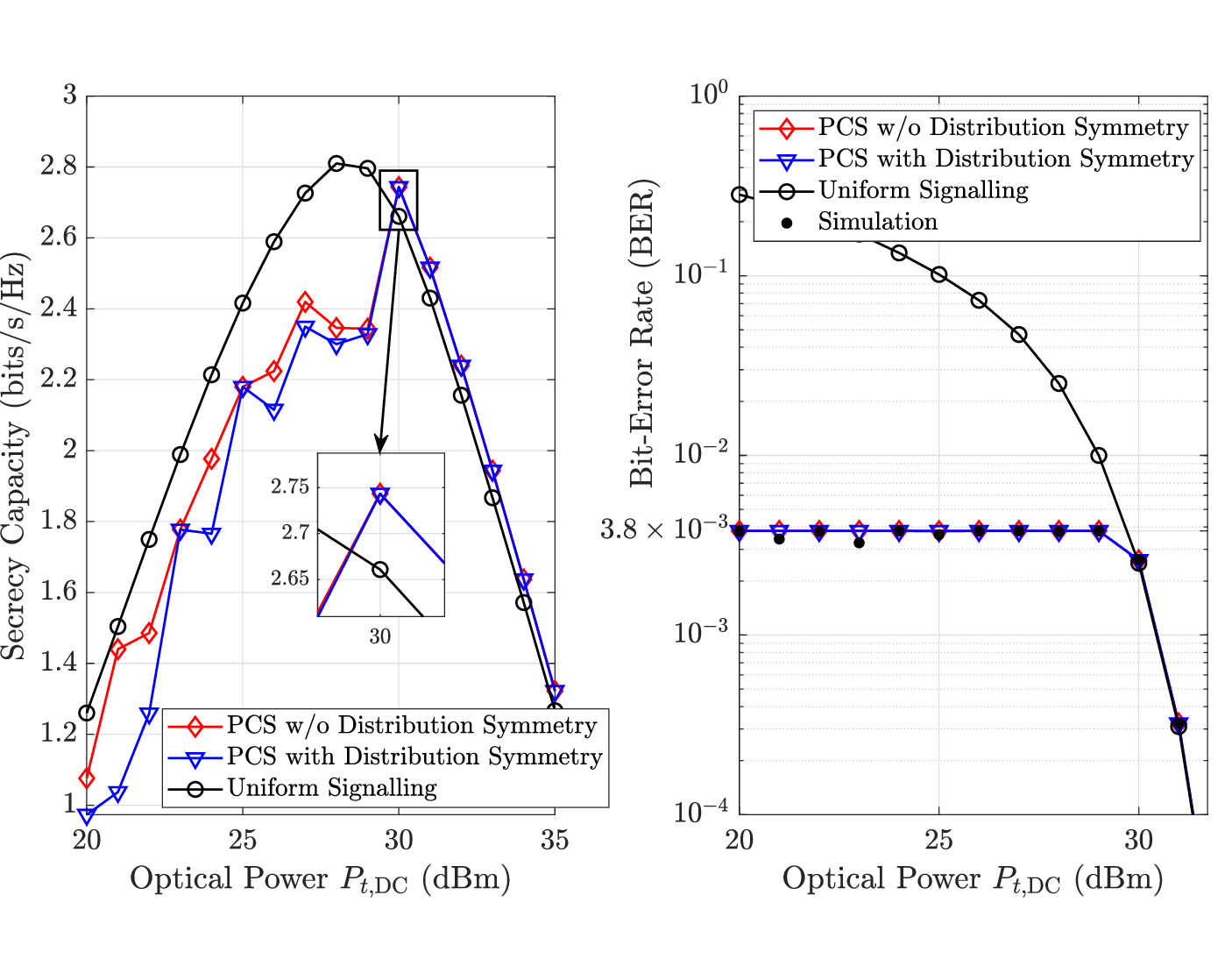}
         \caption{16-PAM.}
         \label{16PAM_Comparison_KnownEve}
     \end{subfigure}
     \hfill
     \caption{Comparison between PCS designs and uniform signaling: Known Eve's CSI.}
     \centering
     \begin{subfigure}[b]{0.49\textwidth}
         \centering
         \includegraphics[width=.9\textwidth, height = 8cm]{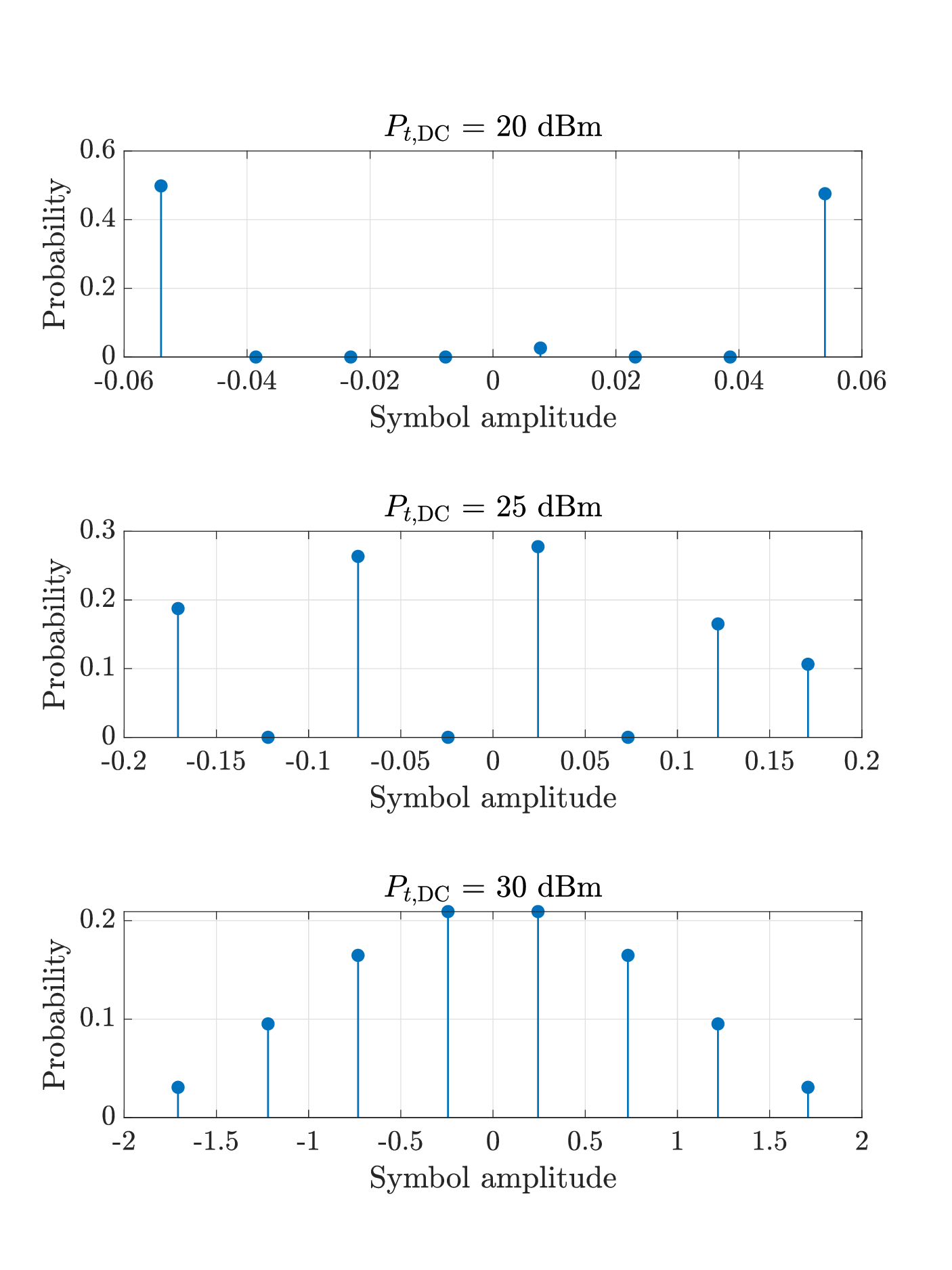}
         \vspace{-5mm}
         \caption{PCS design without distribution symmetry constraint.}
         \label{distribution_PCS}
     \end{subfigure}
     \hfill
     \begin{subfigure}[b]{0.49\textwidth}
         \centering
         \includegraphics[width=.9\textwidth,height = 8cm]{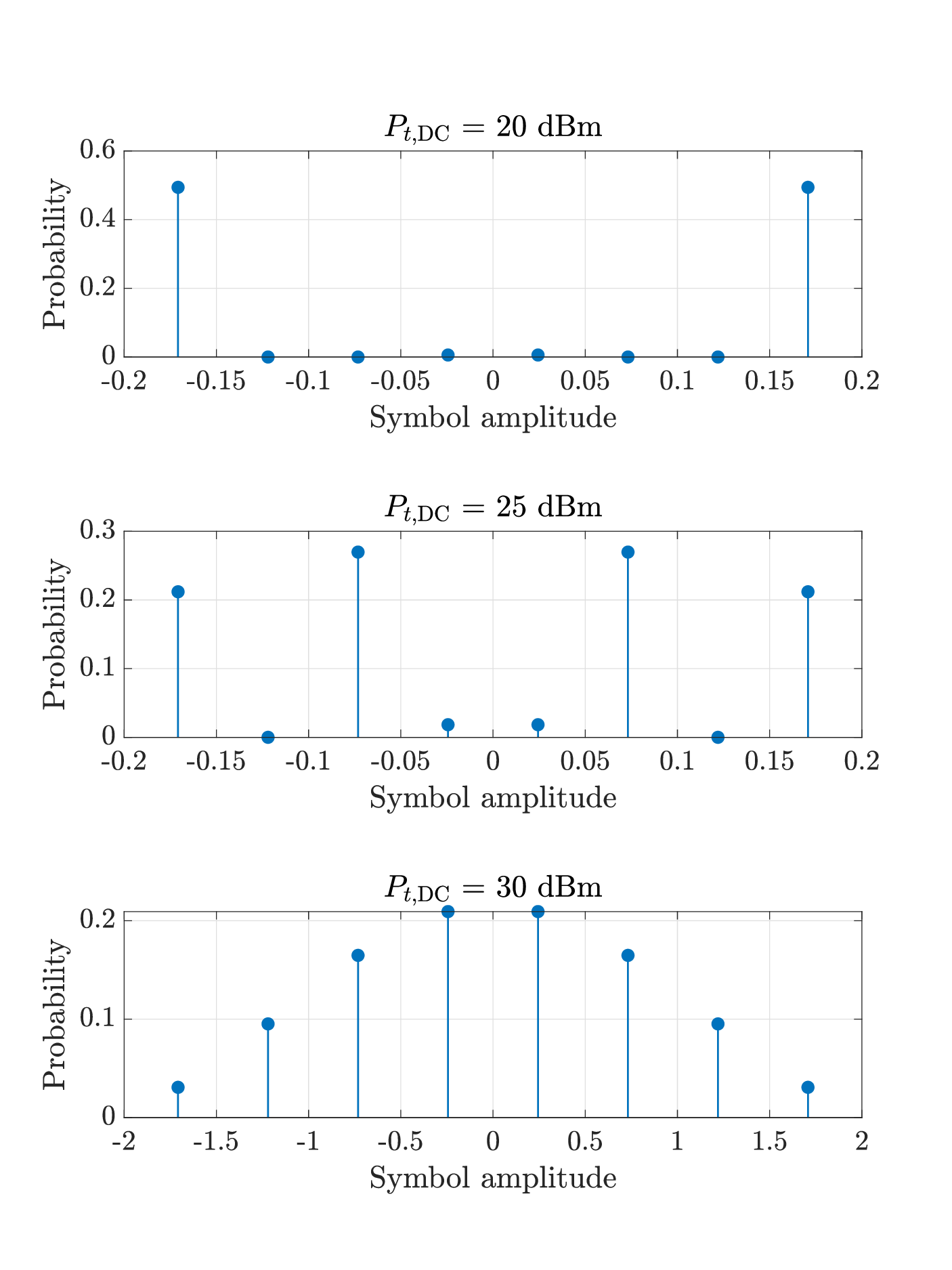}
         \vspace{-5mm}
         \caption{PCS design with distribution symmetry constraint.}
         \label{distribution_PCS_wo}
     \end{subfigure}
     \hfill
     \caption{Optimal constellation distributions of 8-PAM PCS designs: Known Eve's CSI.}
\end{figure*}

For the case of known Eve's CSI, comparisons in terms of the secrecy capacity and BER between PCS designs and the uniform signaling are presented in Fig.~\ref{8PAM_Comparison_KnownEve} and \ref{16PAM_Comparison_KnownEve} for 8-PAM and 16-PAM, respectively. PCS designs with and without the requirement of distribution symmetry are both examined. To guarantee a positive secrecy capacity, it is assumed that $\frac{h_{\text{B}}}{\sigma_{\text{B}}} = 10\frac{h_{\text{E}}}{\sigma_{\text{E}}}$ is chosen. 
Monte-Carlo simulations are also performed to verify the validity of the approximate BER expression in \eqref{PCS-BER}. In the case of uniform signaling, the BER is given by the exact closed-form expression in \cite{Kyongkuk2002}. 
It is clearly shown that there exists a critical value of the optical power $P_{t, \text{DC}}$ (approximately 26.5 dBm and 29.5 dBm for 8-PAM and 16-PAM, respectively) below which the secrecy capacities of the PCS designs is lower than that of the uniform signaling. Obviously, this lower secrecy performance of the proposed PCS designs is due to the introduction of the BER, flickering, and distribution symmetry constraints, which results in a more stringent search space for the optimal distribution $\mathbf{p}$ (at least, when the optical power is not sufficiently high). 
However, it would be emphasized that the uniform signaling violates the pre-FEC BER constraint (the common pre-FEC BER threshold $P_{\text{e}}^{\text{pre}} = 3.8 \times 10^{-3}$ is set \cite{Jian2024}) when $P_{t, \text{DC}}$ is below its critical value. In contrast, the proposed designs effectively limit the BER below the targeted threshold in the entire range of the optical power. When $P_{t, \text{DC}}$ exceeds the critical value, while both uniform signaling and PCS designs satisfy the BER constraint, the latter achieves a better secrecy capacity. For instance,  4.73\% and 3.11\% improvements are attained at $P_{t, \text{DC}} = 30$ dBm for 8-PAM and 16-PAM, respectively.   
As already mentioned, practical VLC systems are subject to a maximum allowable emitted optical power (e.g., due to hardware constraints or to comply with eye safety regulations) and should be able to support dimming control. In other words, the systems should be able to operate under arbitrary adjustment of the optical power within a practical range (e.g., 20 to 35 dBm in our simulations) to meet the required illumination. Hence, the proposed PCS designs with a BER constraint are imperative to guarantee the communication reliability of the legitimate user's channel given any possible practical value of the optical power. 

The optimal 8-PAM symbol distributions of the PCS designs with and without distribution symmetry at $P_{t, \text{DC}} = 20$, $25$, $30$ dBm are shown in Figs.~\ref{distribution_PCS} and \ref{distribution_PCS_wo}, respectively.
As the optical power decreases, it is seen that the number of inactive symbols (i.e., symbols with zero transmission probability) increases, and non-zero transmission probability symbols tend to be far apart. This behavior is the result of increasing symbol susceptibility to noise as the optical power decreases. Hence, to meet the BER constraint, the transmitted symbols (i.e., non-zero transmission probability) should be placed as far as possible to reduce the effect of noise. Note that when the the BER constraint is relaxed, (i.e., when the optical power is sufficiently high, e.g.,  $P_{t, \text{DC}} = 30$ dBm),  designs with and without distribution symmetry result in the same optimal constellation distribution, which resembles a Gaussian-like distribution.  

The performances of PCS designs and uniform signaling for the case of unknown Eve's CSI are shown in Figs.~\ref{8PAM_Comparison_UnknownEve} and \ref{16PAM_Comparison_UnknownEve} where the same behavior as in the case of known Eve's CSI is observed. It is noted, however, that when $P_{t, \text{DC}}$ is beyond its critical value, the average secrecy capacities of PCS designs and uniform signaling are almost the same. Nonetheless, it should be noted that the presented average secrecy capacities of PCS designs are the lower bound of the actual ones, which is given in \eqref{lower-bound}. As also observed in the case of known Eve's CSI, imposing distribution symmetry to the PCS design results in a worse secrecy capacity, especially in the low optical power region. This is because the distribution symmetry constraint, which implies no flickering as a byproduct, is more stringent than that of the flickering mitigation.
Regarding the optimal constellation in this case, it is observed in Fig.~\ref{distribution_PCS_UnknownEve_wo_symmetry} and \ref{distribution_PCS_UnknownEve_with_symmetry} that while the number of inactive symbols also decreases as the optical power decreases, unlike the case of known Eve's CSI, at $P_{t, \text{DC}} = 30$ dBm, instead of a Gaussian-like distribution, the constellation distributions appear to be more uniform. This explains the negligible difference in the secrecy capacity between the uniform signaling and the PCS designs. 
\begin{figure*}
     \centering
     \begin{subfigure}[b]{0.49\textwidth}
         \centering
         \includegraphics[height = 5.8cm, width = \textwidth]{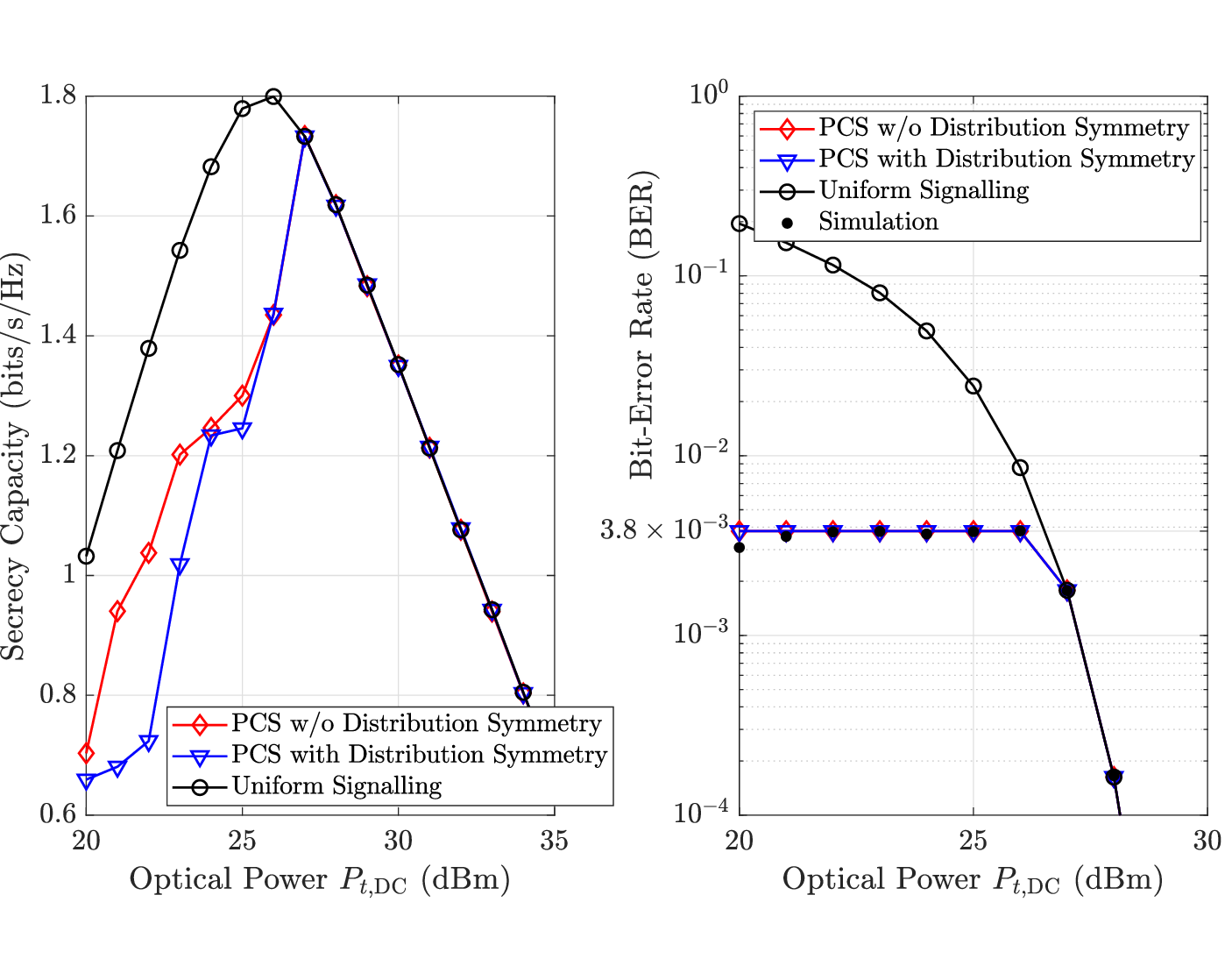}
         \caption{8-PAM.}
         \label{8PAM_Comparison_UnknownEve}
     \end{subfigure}
     \hfill
     \begin{subfigure}[b]{0.49\textwidth}
         \centering
         \includegraphics[height = 5.8cm, width = \textwidth]{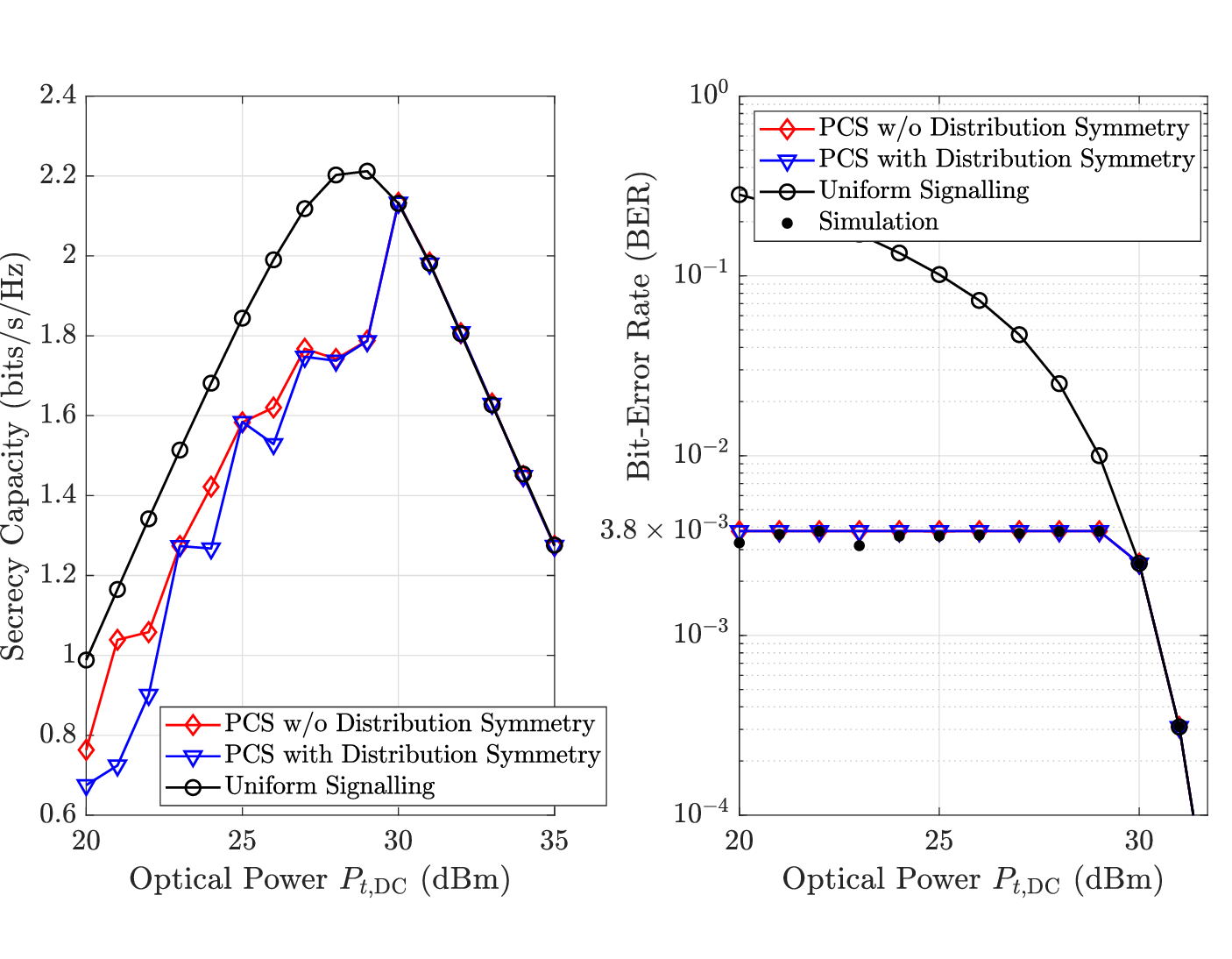}
         \caption{16-PAM.}
         \label{16PAM_Comparison_UnknownEve}
     \end{subfigure}
     \hfill
     \caption{Comparison between PCS designs and uniform signaling: Unknown Eve's CSI.}
     \centering
     \begin{subfigure}[b]{0.49\textwidth}
         \centering
         \includegraphics[width=.9\textwidth, height = 8cm]{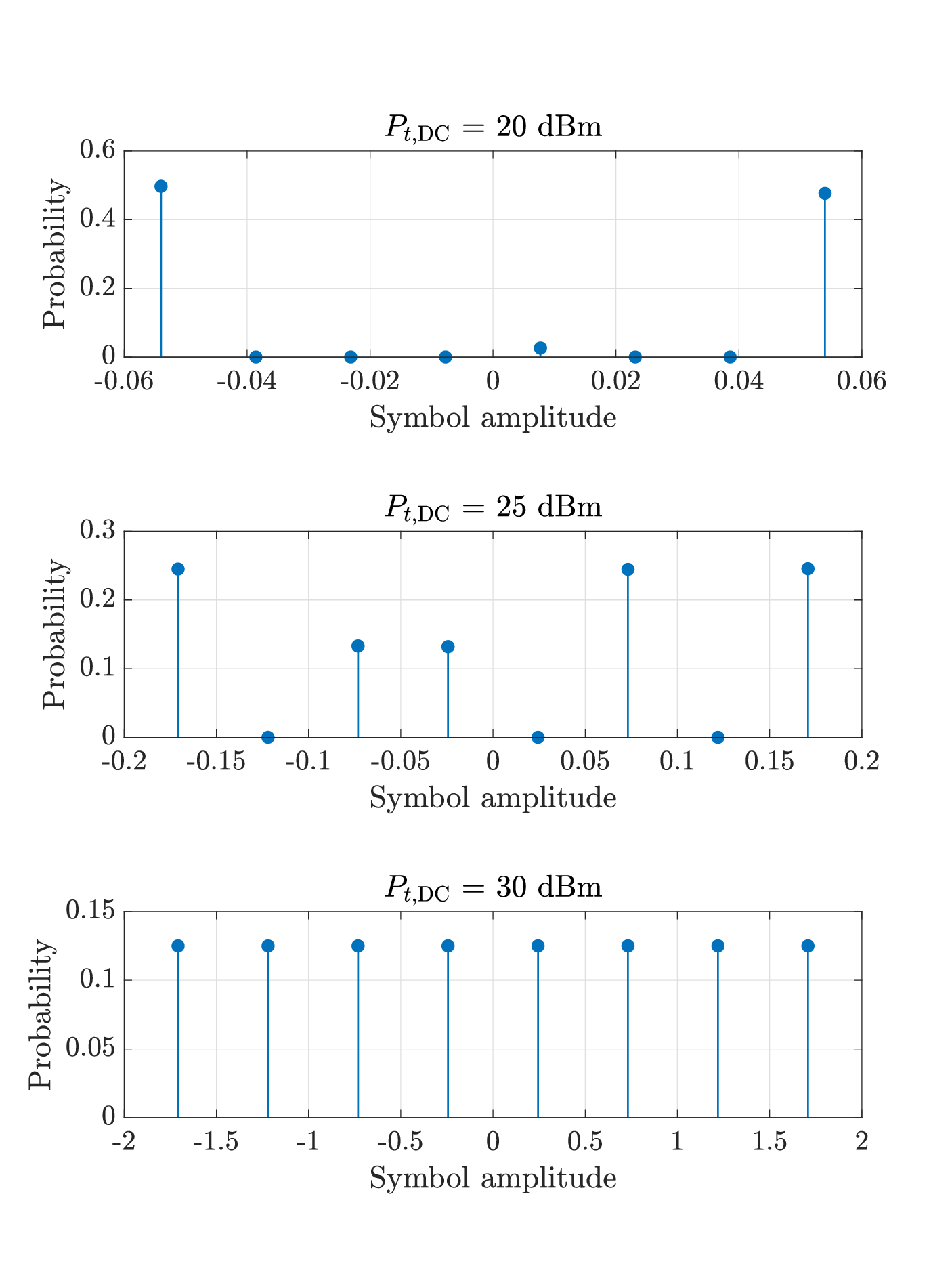}
         \vspace{-5mm}
         \caption{PCS design without distribution symmetry constraint.}
         \label{distribution_PCS_UnknownEve_with_symmetry}
     \end{subfigure}
     \hfill
     \begin{subfigure}[b]{0.49\textwidth}
         \centering
         \includegraphics[width=.9\textwidth,height = 8cm]{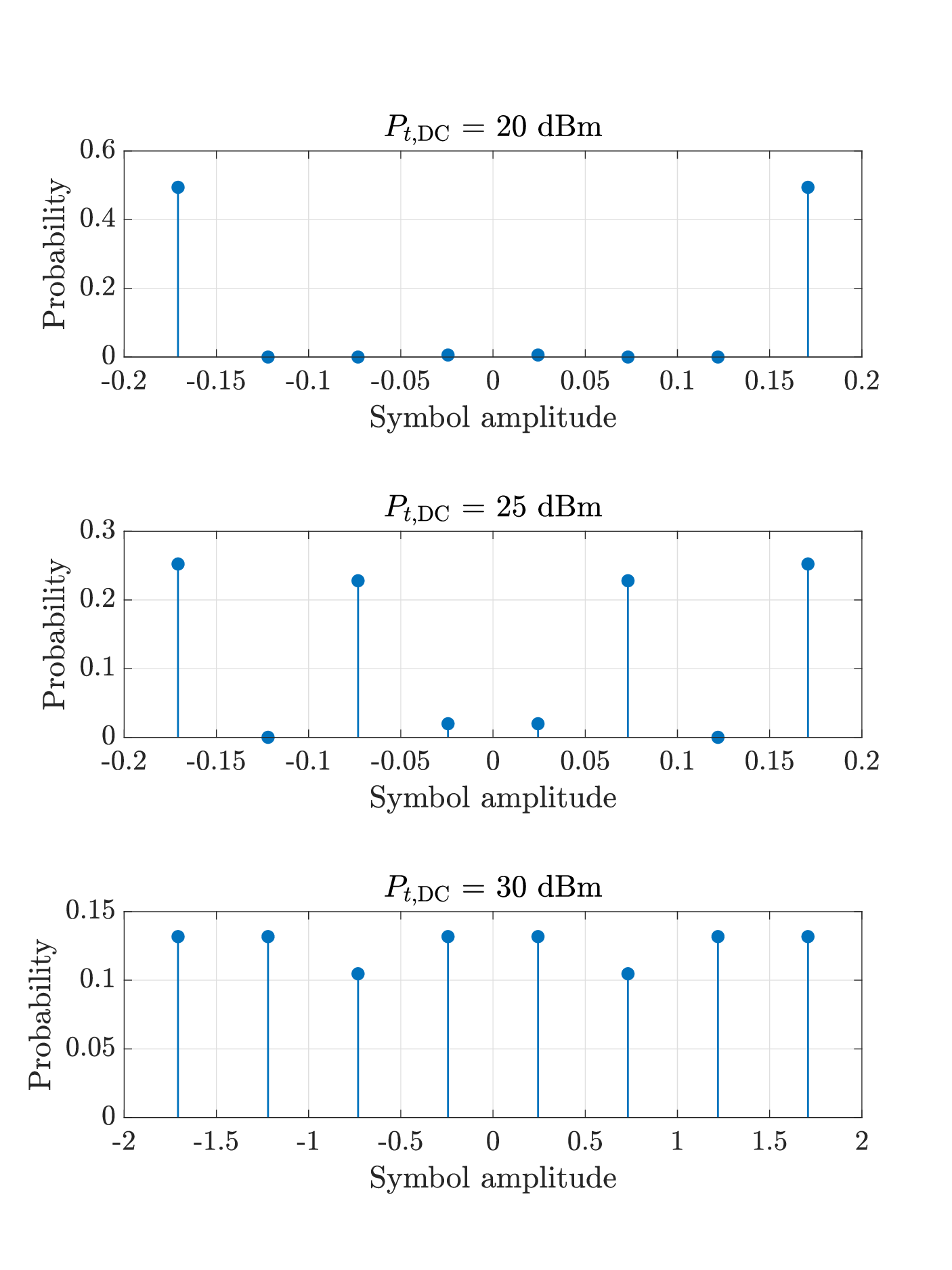}
         \vspace{-5mm}
         \caption{PCS design with distribution symmetry constraint.}
         \label{distribution_PCS_UnknownEve_wo_symmetry}
     \end{subfigure}
     \hfill
     \caption{Optimal constellation distributions of 8-PAM PCS designs: Unknown Eve' CSI.}
\end{figure*}

\begin{figure}[ht]
     \centering
     \begin{subfigure}[b]{0.24\textwidth}
         \centering
         \includegraphics[height = 5.8cm, width = \textwidth]{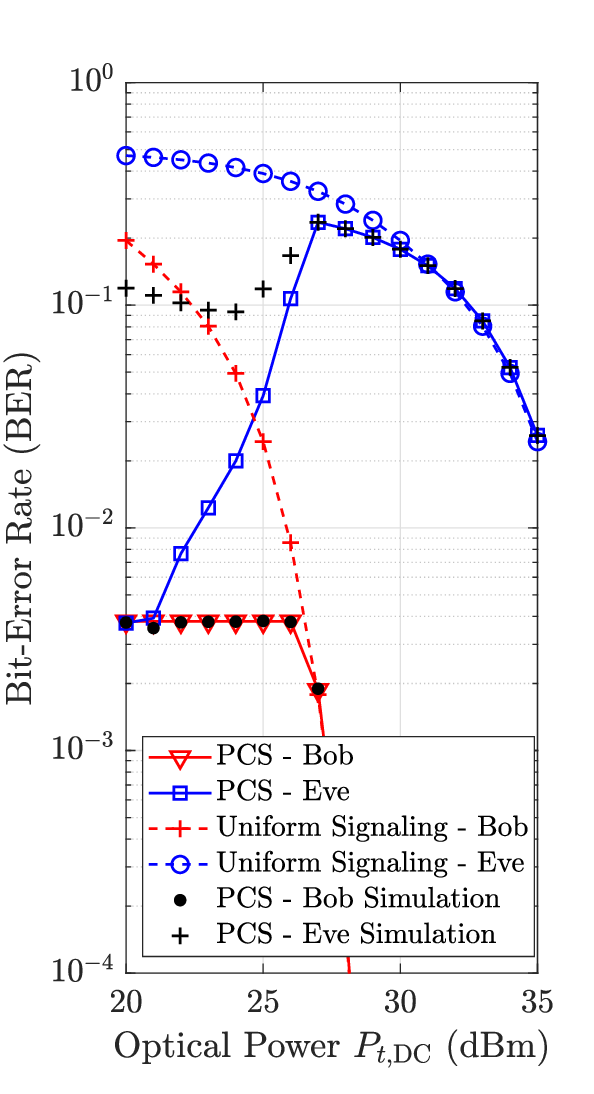}
         \caption{Known Eve's CSI.}
         \label{QoS_KnownEve}
     \end{subfigure}
     \hfill
     \begin{subfigure}[b]{0.24\textwidth}
         \centering
         \includegraphics[height = 5.8cm, width = \textwidth]{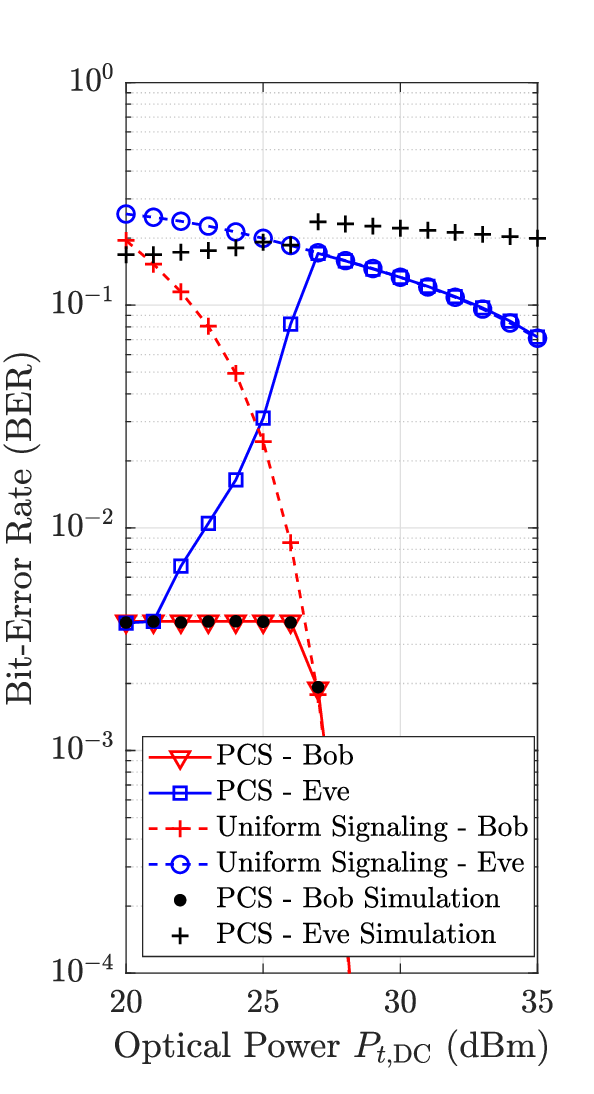}
         \caption{Unknown Eve's CSI.}
         \label{QoS_UnKnownEve}
     \end{subfigure}
     \hfill
     \caption{Performance of QoS-based PCS designs.}
\end{figure}
The BER performances of the QoS-based PCS designs are illustrated in Figs.~\ref{QoS_KnownEve} and \ref{QoS_UnKnownEve} for the case of known and unknown Eve's CSI, respectively. Here, we consider the 8-PAM constellation. For this design approach, we pay attention to the BER of Eve's channel, which determines the level of physical security as the higher BER of Eve's channel implies that less information is correctly decoded by Eve. We observe that as the optical power increases, the approximate BER of Eve's channel (i.e., the solution of problem \eqref{OptProb8}) increases until a certain value then starts decreasing. This is because increasing the optical power leads to more relaxed constraints, resulting in a higher BER of Eve's channel (i.e., better objective value). Nonetheless, when the power exceeds a certain threshold, its impact becomes dominant; thus, as the power increases, the BER decreases. It is noticed that the actual BER of Eve's channel depicted by simulations is much higher than the pre-FEC BER threshold, which renders the unreliability of Eve's channel, thus confirming the effectiveness of our proposed QoS-based PCS designs. 

\section{Conclusion} 
In this paper, practical PCS designs of $M$-PAM constellation to improve the PLS performance of VLC channels have been investigated considering constraints on channel reliability, flickering, and distribution symmetry. To incorporate the channel reliability into the PCS design problems, tractable closed-form expressions for the lower and upper bound BER of PCS-PAM under arbitrary constellation distribution were derived. By proving the concavity of the lower and upper bound BER expressions, the proposed PCS design problems for both cases of known and unknown Eve's CSI were solved using the CCCP algorithm. 
Numerical simulations verified the effectiveness of the proposed designs as the considered constraints were met under the entire range of the practical LED's emitted optical power. Furthermore, it was also revealed that the per-FEC BER constraint renders the optimal constellation symbols sparse at the low optical power regime. For future work, one may consider PCS designs for systems with multiple legitimate users and/or multiple eavesdroppers where the constellation distributions for the users should be jointly optimized.  
\begin{appendices}
\vspace{-3mm}
\section{Proof of Proposition 1}
We rewrite \eqref{SEP1} as
\begin{align}
    P^{\text{U}}_{m, n}(\mathbf{p}) &= \text{Pr}\left(\frac{p_m}{p_n} \leq \frac{p(\overline{y}_{\text{U}}|s = s_n)}{p(\overline{y}_{\text{U}}|s = s_m)}\right) \nonumber \\
    & = \text{Pr}\left(\log\frac{p_m}{p_n} \leq -\frac{(\overline{y}_{\text{U}} - r_{\text{U}, n})^2}{2\sigma^2_{\text{U}}} + \frac{(\overline{y}_{\text{U}} - r_{\text{U}, m})^2}{2\sigma^2_{\text{U}}}\right).
\end{align}
Since we assume that the symbol $m$ is transmitted, we have $\overline{y}_{\text{U}} = r_{\text{U},m} + n_{\text{U}}$. Therefore, 
\begin{align}
    P^{\text{U}}_{m, n}(\mathbf{p}) &= \text{Pr}\left(2\sigma^2_{\text{U}}\log\frac{p_m}{p_n} \leq -\left(r_{\text{U}, m} - r_{\text{U}, n} + n_{\text{U}}\right)^2 + n^2_{\text{U}}\right) \nonumber \\
    & = \text{Pr}\left(n_{\text{U}}d_{\text{U}, m, n} \leq \frac{-2\sigma^2_{\text{U}}\log\frac{p_m}{p_n} - d^2_{\text{U}, m, n}}{2}\right) \nonumber \\
    & \overset{(a)}{=} \frac{1}{2}\text{erfc}\left(\frac{2\sigma^2_{\rm{U}}\log\frac{p_m}{p_n} + d^2_{{\rm{U}}, m, n}}{2\sqrt{2}\sigma_{{\rm{U}}}|d_{{\rm{U}}, m, n}|}\right),
    \label{prob11}
\end{align}
where $d_{\text{U}, m, n} = r_{\text{U},m} - r_{\text{U}, n} = h_{\text{U}}\gamma\eta(a_m - a_n)$ and $(a)$ is due to the fact that  $n_{\text{U}}d_{\text{U}, m, n}$ is a Gaussian random variable with mean 0 and variance $\sigma^2_\text{U}d^2_{\text{U}, m, n}$. The proof is completed. 
\section{Proof of Proposition 2}
We rewrite \eqref{PCS-BER} as 
\begin{align}
    P^{\text{U}}_{\text{e,ub}}(\mathbf{p}) & \approx \frac{1}{\log_2M}\sum_{m = 1}^M\sum_{\substack{n = 1 \\ n\neq m}}^M\left(p_m{\rm{erfc}}\left(\frac{2\sigma^2_{\rm{U}}\log\frac{p_m}{p_n} + d^2_{{\rm{U}}, m, n}}{2\sqrt{2}\sigma_{{\rm{U}}}|d_{{\rm{U}}, m, n}|}\right) \right. \nonumber \\ & \hspace{23mm}\left. + p_n{\rm{erfc}}\left(\frac{2\sigma^2_{\rm{U}}\log\frac{p_n}{p_m} + d^2_{{\rm{U}}, n, m}}{2\sqrt{2}\sigma_{{\rm{U}}}|d_{{\rm{U}}, n, m}|}\right) \right) \nonumber \\ 
    & = \frac{1}{\log_2M}\sum_{m = 1}^M\sum_{\substack{n = 1 \\ n\neq m}}^M g_{m, n}(\mathbf{p}).
\end{align}
To show that $P^{\text{U}}_{\text{e,ub}}(\mathbf{p})$ is concave, it suffices to prove that for any $m \neq n $, the following function 
\begin{align}
    g_{m, n}(\mathbf{p}) = & p_m{\rm{erfc}}\left(\frac{2\sigma^2_{\rm{U}}\log\frac{p_m}{p_n} + d^2_{{\rm{U}}, m, n}}{2\sqrt{2}\sigma_{{\rm{U}}}|d_{{\rm{U}}, m, n}|}\right) \nonumber \\ &+ p_n{\rm{erfc}}\left(\frac{2\sigma^2_{\rm{U}}\log\frac{p_n}{p_m} + d^2_{{\rm{U}}, n, m}}{2\sqrt{2}\sigma_{{\rm{U}}}|d_{{\rm{U}}, n, m}|}\right)
\end{align}
is concave over $\mathbf{p}$. For the sake of conciseness, let $\alpha_{m, n}  = \frac{\sigma_{\text{U}}}{\sqrt{2}|d_{\text{U}, m, n}|}$, $\alpha_{n, m}  = \frac{\sigma_{\text{U}}}{\sqrt{2}|d_{\text{U}, n, m}|}$, $\beta_{m, n}  = \frac{|d_{\text{U}, m, n}|}{2\sqrt{2}\sigma_{\text{U}}}$, and $\beta_{n, m}  = \frac{|d_{\text{U}, n, m}|}{2\sqrt{2}\sigma_{\text{U}}}$. Note that since $|d_{\text{U}, m, n}| = |d_{\text{U}, n, m}|$, $\alpha_{m, n}  = \alpha_{n, m}  = \alpha > 0$ and $\beta_{m, n}  = \beta_{n, m} = \beta > 0$. Also, let $\theta_{m, n} = \alpha \log\frac{p_m}{p_n} + \beta $ and $\theta_{n, m} = \alpha \log\frac{p_n}{p_m} + \beta$. The Hessian of $g(\mathbf{p})$ is 
\begin{align}
    \mathbf{H}_{g, m, n} = \begin{bmatrix}
        0 & \cdots & 0 & \cdots & 0 & \cdots & 0 \\
        \vdots & & \vdots & & \vdots & & \vdots\\
        0 & \cdots & \frac{\partial^2g_{m, n}(\mathbf{p})}{\partial p^2_m} & \cdots &  \frac{\partial^2g_{m, n}(\mathbf{p})}{\partial p_m\partial p_n} & \cdots & 0 \\
        \vdots & & \vdots & & \vdots & & \vdots \\ 
        0 & \cdots \ & \frac{\partial^2g_{m, n}(\mathbf{p})}{\partial p_n\partial p_m} & \cdots & \frac{\partial^2g_{m, n}(\mathbf{p})}{\partial p^2_n} & \cdots & 0 \\
        \vdots & & \vdots & & \vdots & & \vdots \\
        0 & \cdots & 0 & \cdots & 0 & \cdots & 0 \\
    \end{bmatrix},
\end{align}
where 
\begin{align}
    \frac{\partial^2 g_{m, n}(\mathbf{p})}{\partial p^2_m} = \frac{2}{\sqrt{\pi}p^2_m}\left(2\alpha^2p_m \exp(-\theta^2_{m, n})\theta_{m, n} - \alpha p_m\exp(-\theta^2_{m, n}) \right. \nonumber \\ \left. + 2\alpha^2p_n\exp(-\theta^2_{n, m})\theta_{n, m} - \alpha p_n\exp(-\theta^2_{n, m})\right), \nonumber \\
\end{align}
\begin{align}
    \frac{\partial^2 g_{m, n}(\mathbf{p})}{\partial p^2_n} = \frac{2}{\sqrt{\pi}p^2_n}\left(2\alpha^2p_m \exp(-\theta^2_{m, n})\theta_{m, n} - \alpha p_m\exp(-\theta^2_{m, n}) \right. \nonumber \\ \left. + 2\alpha^2p_n\exp(-\theta^2_{n, m})\theta_{n, m} - \alpha p_n\exp(-\theta^2_{n, m})\right),
\end{align}
and
\begin{align}
    &\frac{\partial^2 g_{m, n}(\mathbf{p})}{\partial p_m\partial p_n}  = \frac{\partial^2 g_{m, n}(\mathbf{p})}{\partial p_n\partial p_m} \nonumber \\ &= \frac{2}{\sqrt{\pi}p_mp_n}\left(-2\alpha^2p_m \exp(-\theta^2_{m, n})\theta_{m, n} + \alpha p_m\exp(-\theta^2_{m, n}) \right. \nonumber \\ & ~~~~~\left.  - 2\alpha^2p_n\exp(-\theta^2_{n, m})\theta_{n, m} + \alpha p_n\exp(-\theta^2_{n, m})\right).
\end{align}
For any nonzero real vector $\mathbf{z} = \begin{bmatrix}
    z_1 & z_2 & \cdots & z_M
\end{bmatrix}^T$, we have
\begin{align}
    &\mathbf{z}^T \mathbf{H}_{g, m, n} \mathbf{z} = \nonumber \\ & \frac{2}{\sqrt{\pi}}\left(\frac{z_m}{p_m} - \frac{z_n}{p_n}\right)^2 \nonumber  \left(2\alpha^2p_m \exp(-\theta^2_{m, n})\theta_{m, n}  - \alpha p_m\exp(-\theta^2_{m, n}) \right. \nonumber \\ & \left. + 2\alpha^2p_n\exp(-\theta^2_{n, m})\theta_{n, m} - \alpha p_n\exp(-\theta^2_{n, m})\right).
\end{align}
Let $\Phi_{m, n} = 2\alpha^2p_m \exp(-\theta^2_{m, n})\theta_{m, n}  - \alpha p_m\exp(-\theta^2_{m, n})  + 2\alpha^2p_n\exp(-\theta^2_{n, m})\theta_{n, m} - \alpha p_n\exp(-\theta^2_{n, m})$. To prove the concavity of $g_{m, n}(\mathbf{p})$, it is necessary to show that  $\Phi_{m, n}$ is non-positive. Indeed, we have
\begin{align}
    \frac{\Phi_{m, n}}{\exp(-\theta^2_{m, n})}  & =  2\alpha^2p_m\theta_{m, n} + 2\alpha^2p_n\exp(\theta^2_{m, n} - \theta^2_{n, m})\theta_{n, m} \nonumber \\ & ~~~- \alpha p_m - \alpha p_n\exp(\theta^2_{m, n} - \theta^2_{n, m}).
    \label{phi_m, n}
\end{align}
Notice that $\exp(\theta^2_{m, n} - \theta^2_{n, m}) = \exp((\theta_{m, n}-\theta_{n, m})(\theta_{m, n} + \theta_{n, m})) = \exp\left(4\alpha\beta\log\frac{p_m}{p_n}\right) = \frac{p_m}{p_n}$. The expression in \eqref{phi_m, n} becomes
\begin{align}
    \frac{\Phi_{m, n}}{\exp(-\theta^2_{m, n})} & =   2\alpha^2p_m(\theta_{m, n} + \theta_{n, m}) - 2\alpha p_m \nonumber \\ & = 4\alpha^2\beta p_m - 2\alpha p_m \nonumber \\ & = \alpha p_m -2 \alpha p_m  = -\alpha p_m \leq 0. 
\end{align}
Therefore, we have $\mathbf{z}^T \mathbf{H}_{g, m, n} \mathbf{z} \leq 0$, which implies that $g_{m, n}(\mathbf{p})$ is concave. Since $ P^{\text{U}}_{\text{e,ub}}(\mathbf{p})$ is the sum of concave functions, it is also concave. The proof is complete.
\end{appendices}
\bibliographystyle{IEEEtran}
\bibliography{references}
\end{document}